\documentclass[journal]{IEEEtran}
\IEEEoverridecommandlockouts

\usepackage[cmex10]{amsmath}
\usepackage{amsfonts}
\usepackage{cellspace}
\usepackage{makecell}
\setcellgapes{3.5pt}
\usepackage[caption=false,font=footnotesize]{subfig}
\usepackage{xcolor}
\usepackage{booktabs}
\usepackage{multirow}
\usepackage{makecell}
\usepackage{soul}
\usepackage{cite}
\usepackage{amssymb}
\usepackage{amsthm}

\newtheorem{remark}{Remark}
\newtheorem{claim}{Claim}

\usepackage{breqn}
\usepackage{algorithm} 
\usepackage{algpseudocode}

\usepackage{pstricks}
\newcommand{\norm}[1]{\lVert #1 \rVert}

\newcommand{\mads}[1]{{\red(Mads:)}}

\begin{document}

\title{Maximizing Power Flexibility of Hybrid Energy Systems for Capacity Market\\

\thanks{This material is based upon work supported by the U.S. Department of Energy’s Office of Energy Efficiency and Renewable Energy (EERE) under the Solar Energy Technologies Office Award Number DE-EE0010147. The views expressed herein do not necessarily represent the views of the U.S. Department of Energy or the United States Government.

Both the authors are with the Department of Electrical and Biomedical Engineering, University of Vermont, Burlington, VT 05405 USA. (emails: tmishra@uvm.edu, malmassa@uvm.edu)}
}

\author{
\IEEEauthorblockN{Tanmay Mishra, \IEEEmembership{Member, IEEE}, Mads R Almassalkhi, \IEEEmembership{Senior Member, IEEE}}
}
\maketitle
\vspace{-1.0cm}
\begin{abstract}
Hybrid Energy Systems (HES), integrating generation sources, energy storage, and controllable loads, are well-positioned to provide real-time grid flexibility. However, quantifying this maximum flexibility is challenging due to renewable generation uncertainty and the complexity of power allocation across multiple assets in real time. This paper presents a rule-based framework for characterizing HES flexibility and systematically allocating power among its constituent assets. The flexibility envelope defines the dynamic power boundary within which the HES can inject or absorb power without violating operational constraints. Shaped in real time by capacity bids, available solar generation, and power allocation protocol, it enables reliable and predictable HES participation in regulation markets. Depending on the operational objective, the framework supports both symmetric and asymmetric flexibility cases. Further, the proposed power-allocation rule is benchmarked against an optimal dispatch, providing a performance reference under realistic conditions. Finally, state of charge drift correction control is presented to ensure sustained battery operation and system reliability. This work, therefore, offers a rigorous and practical framework for integrating HES into capacity markets through effective flexibility characterization.
\end{abstract}

\vspace{-0.4cm}
\section{Introduction}
Modern power systems are shifting away from deterministic, dispatchable centralized power plants to more decentralized and variable energy sources, largely driven by evolving energy economics and decarbonization goals. Although this transition brings environmental benefits, it also introduces new challenges in maintaining grid planning and operation. For example, California's well-known `duck curve' highlights the requirement of balancing demand and generation during peak solar hours~\cite{Calero2022}. In Hawaii, high solar penetration has led to frequent curtailments simply because the system lacks the flexibility to absorb all of the generation~\cite{oshaughnessy2020}. Even advanced battery systems, often touted as a solution, have their limits. During the 2019 UK blackout, batteries were unable to provide the sustained support needed, exposing gaps in how standalone technologies handle complex disturbances~\cite{BIALEK2020}. 
In recent years, the rapid growth of data centers and their escalating power demands make it increasingly challenging for a single generation resource to synchronize their production with the consumption of the data centers \cite{Gnibga2024}.

This presents an opportunity for coordinated systems that can handle variability and offer flexible control over different timescales to make the grid more resilient~\cite{LiHao2018}. The critical need here is for grid flexibility, defined as the ability of the power system to maintain reliable operation despite variability and uncertainty in both electricity generation and demand. Hybrid energy systems (HES) give us that option by combining different energy resources in form of generation, storage, and controllable loads, which can respond autonomously to both sudden events and longer mismatches between supply and demand~\cite{Murphy2021, Tanmay2024}. 

Historically, grid flexibility was primarily ensured through large, dispatchable fossil fuel-based generators augmented by spinning reserves and hydropower ~\cite{palensky2011demand}. While these centralized resources provided reliable control on minute-to-hour timescales, their dependence on a few large units constrains the system’s flexibility in addressing the variability and uncertainty of today’s decentralized energy landscape. In recent years, the integration of distributed energy resources (DERs) and the emergence of virtual power plants (VPPs) represent a new dynamic and scalable solution to power system operations~\cite{Xie2020,Brahma2021,Brahma2023}. HES build on this evolution by providing rapid response via batteries and enhanced operational flexibility through the integration of diverse assets, including solar PV, controllable loads such as data centers~\cite{Wang2016}, EV charging stations~\cite{Karfopoulos2015} and process-driven loads like hydrogen electrolyzer~\cite{Sovana2024}, enabling both short-term adjustments and longer-duration energy management. 

While HES offer a promising pathway for grid flexibility, realizing their full potential requires overcoming significant stochastic and operational barriers. First, the intermittency of renewable generation and the varying constraints of controllable loads make it inherently difficult to quantify the maximum flexibility the system can reliably offer~\cite{Mohandes2019,Garcia2021}. Energy market participation further complicates matters, requiring precise forecasting and robust assessment of the HES's regulation capabilities to formulate accurate bids for electricity and ancillary service markets~\cite{mishra2025}. These challenges further creates a complex real-time power allocation challenge. Even though HES presents itself to the grid as a single dispatchable entity, yet internally, the aggregate power request must be in real-time and optimally distributed among diverse assets. Finally, physical sustainability poses a persistent threat to reliability. Even under theoretically energy-neutral regulation signals, round-trip efficiency losses and short-term energy imbalances inevitably cause battery's state of charge (SoC) drift~\cite{Dang2014}.

To overcome these challenges, this article proposes a novel framework to characterize the flexibility of HES through real-time disaggregation strategy that allocate power among individual assets. This real-time power allocation protocol is particularly advantageous for fast-responding markets like frequency regulation, offering the computational efficiency and responses crucial for real-time compliance. In this work, a flexibility envelope is defined quantifying the HES's real-time power dispatch capacity. This envelope is defined by integrating operational rules and PV curtailment. This paper further investigates the impact of solar power generation variability across different seasons on HES performance, and evaluates its behavior under different bid capacities using performance matrices. Finally, an adaptive SoC drift correction approach is proposed to ensure sustainable operation within battery limits during multi-hour participation.
 
The remainder of the paper is organized as follows. Section~\ref{system} presents the mathematical modeling of the HES components. Section~\ref{Flexibility} introduces the flexibility envelope concept and defines the characterization of maximum HES capacity alongside the real-time power allocation protocol. Further, an optimization framework is formulated to benchmark the rule-based dispatch against the optimal power dispatch. Section~\ref{sec:market_setting} describes the frequency regulation market setting. Section~\ref{sec:PJM_case_study} evaluates the performance of HES participation in frequency regulation market based on the defined flexibility envelope and proposed dispatch strategy as a case study. This section also introduces the adaptive SoC drift correction strategy designed to maintain operational limits during sustained, multi-hour participation. Finally, Section~\ref{Conclusion} summarizes key findings and outlines future work.

\section{Modeling of HES Components}
\label{system}
Fig.~\ref{HES_schematic} shows an HES connected to a utility grid. In this work, the HES comprises a solar PV as power generation source injecting power $P_{\text{pv}}$, a battery as energy storage with dispatchable power $P_{\text{batt}}$, and a controllable load absorbing power $P_{\text{CL}}$. The system is modeled for all discrete time steps $k \in \mathcal{K} = \{0,1,\dots,N-1\}$ of the operational window, where each $k$ denotes a discrete time index. Each asset $X$ has a bounded power consumption at $k$, denoted by $P_{\text{X}}[k] \in [\underline{P}_{\text{X}}, \overline P_{\text{X}}]$, where $\underline{P}_{\text{X}} = 0$ for $X \in \{\text{CL}, \text{PV}\}$ and $\underline{P}_{\text{X}} = -\overline P_{\text{X}}$ for $X = \text{batt}$. The power balance for the system is given by:
\begin{equation}
    P_{\text{hes}}[k] = P_{\text{PV}}[k] - P_{\text{CL}}[k] + P_{\text{batt}}[k],
    \label{eq:hes_power_bal}
\end{equation}
where $P_{\text{hes}}[k] > 0$ implies that the HES is injecting power into the grid at time $k$. For effective regulation of $P_{\text{hes}}[k]$ through asset coordination, simplified models are employed for each component. These models capture the essential power limits of the PV system, controllable load, and battery (including its state of charge), as described in the following subsections. 
\begin{figure}[t]
\centering
\includegraphics[width=0.48\textwidth]{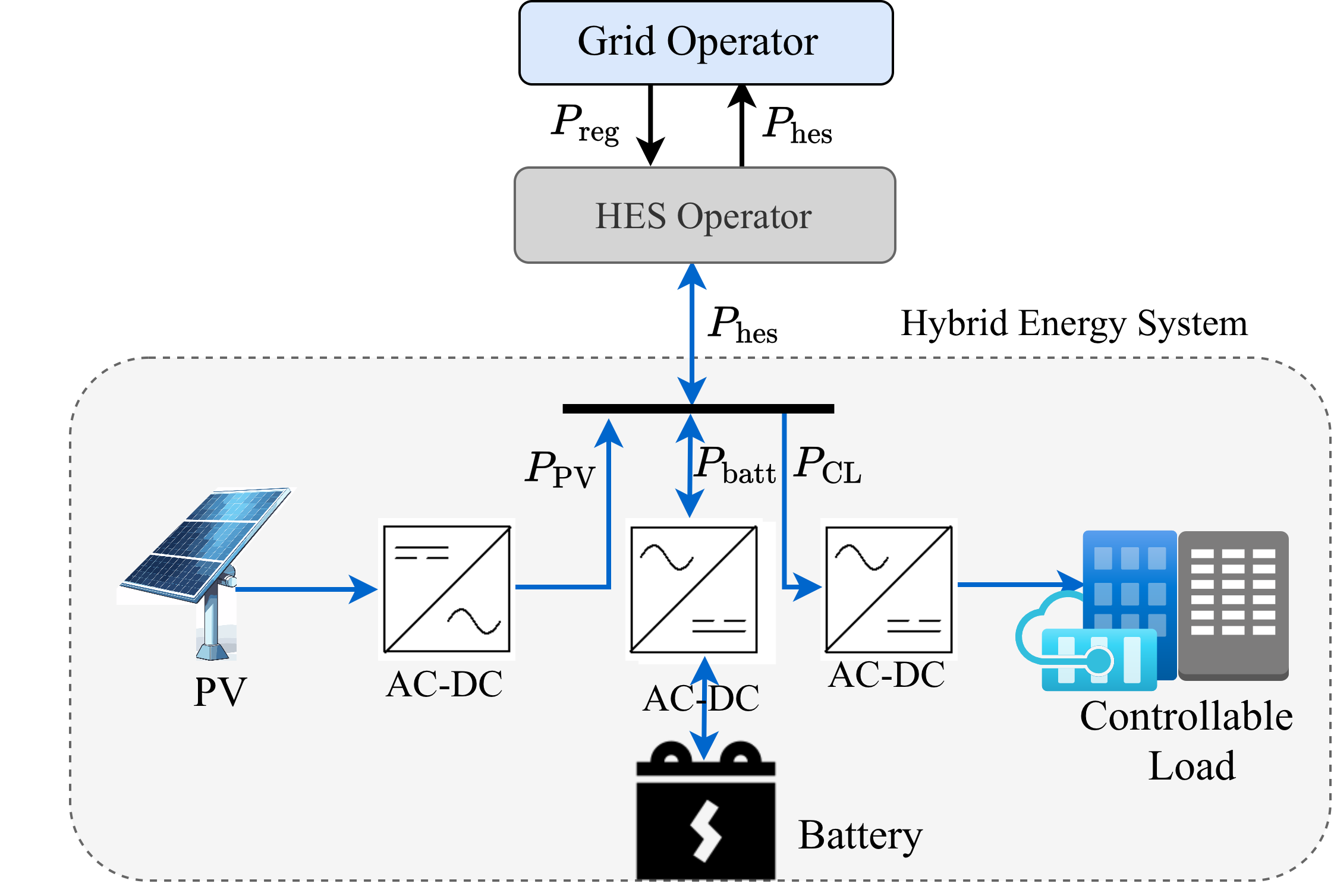}
\caption{Schematic diagram of a hybrid energy system interfacing PV, battery and a controllable load to the grid through inverters.}
\label{HES_schematic}
\centering
\end{figure}

\vspace{-0.25cm}
\subsection{PV model:}
The equivalent circuit model of a PV cell, as described in~\cite{Masters2004}, is used to represent solar power generation using irradiance data. The produced cell current $i_{\text{cell}}[k] $ is given by:
\begin{equation}
    i_{\text{cell}}[k] = i_{\text{sc}}[k]-I_{\text{0}} \left(e^{38.9v^{\text{pv}}_{\text{oc}}[k]}-1\right)-\left(\frac{v^{\text{pv}}_{\text{oc}}[k]}{R_{\text{P}}}\right)
    \label{PV_current}
\end{equation}
\noindent where $i_{\text{sc}}[k]$ is the short circuit current and $v^{\text{pv}}_{\text{oc}}[k]$ is the open circuit terminal voltage. The constant value of 38.9 is calculated at the operating temperature of 300K from the charge and Boltzmann's constant. $R_{\text{P}}$ and $R_{\text{S}}$ denote losses due to leakage and interconnections. The $i_{\text{sc}}[k]$ is proportional to the irradiance $S_{\text{I}}[k]$. The cell voltage $v_{\text{cell}}[k]$ and power $P_{\text{cell}}[k]$ are given by:
\begin{subequations}
\begin{align}
v_{\text{cell}}[k] &= v^{\text{pv}}_{\text{oc}}[k]-i_{\text{cell}}[k]R_{\text{S}}\\
P_{\text{cell}}[k] &= i_{\text{cell}}[k]v_{\text{cell}}[k]\\
P_{\text{PV}}[k] &=  \eta _{\text{PV}}N_{\text{cell}}P_{\text{cell}}[k]
\end{align}
\label{PV_Cell_Power}
\end{subequations}
\noindent where $\eta _{\text{PV}}$ is PV inverter efficiency and $N_{\text{cell}}$ corresponds to the aggregate number of PV cells to represent a scaled PV arrays i.e. bulk PV system, configured through the combination of series connected modules, parallel strings, and the total number of PV arrays. Hence,  $P_{\text{PV}}[k]$ denotes the total solar power generated at each discrete time step $k$ .

\subsection{Battery model:}
The Energy Reservoir Model (ERM) is commonly used in system-level battery dispatch scheduling and flexibility studies~\cite{Rosewater2019, Arash2025}. In this model, the battery's SoC, $E[k+1] \in [\underline{E}, \overline{E}]$, depend on the charging power $P^{\text{c}}_{\text{batt}}[k]$ and discharging power $P^{\text{d}}_{\text{batt}}[k]$,  with the total battery power is given as $P_{\text{batt}}[k] =P^{\text{c}}_{\text{batt}}[k]+P^{\text{d}}_{\text{batt}}[k]$. In this case, charging is treated as power consumption i.e., $P^{\text{c}}_{\text{batt}}[k] \le 0$, whereas discharging is injecting power i.e., $P^{\text{d}}_{\text{batt}}[k] \ge 0$.

The SoC defined by its capacity $E_{\text{C}}$ in MWh. The SoC of the battery is modeled as follows~\cite{Ferreira2019}:
\begin{subequations}
\begin{equation}
E[k+1] = E_{0} - \frac{\Delta t}{E_{\text{C}}} \sum_{j=1}^{k} \left( \eta^{\text{batt}}_{\text{i}} P^{\text{c}}_{\text{batt}}[j] + \frac{1}{\eta^{\text{batt}}_{\text{i}}} P^{\text{d}}_{\text{batt}}[j] \right)
\label{SoC_dynamics_1}
\end{equation}
\begin{equation}
\text{SoC limits}:~ \underline{E} \le E[k+1] \le \overline{E} \label{SoC_limits},
\end{equation}

\begin{equation}
-\overline{P}_{\text{batt}} \le P^{\text{c}}_{\text{batt}}[k] \le 0, ~ 0\le P^{\text{d}}_{\text{batt}}[k] \le \overline{P}_{\text{batt}}\label{P_batt}
\end{equation}
\label{battery model}
\end{subequations}

where \eqref{SoC_limits} and \eqref{P_batt} define the battery’s SoC and charging/discharging limits, respectively. This formulation captures energy conversion losses via the inverter efficiency as $\eta^{\text{batt}}_{\text{i}}$. The efficiency of these inverters is assumed to be constant at $\eta^{\text{batt}}_{\text{i}}=~0.95$, based on typical manufacturer datasheet.

\subsection{Controllable load}
The controllable load models dispatchable demand-side resources such as data centers, process-driven industrial loads (e.g., electrolyzer), or electric vehicle charging stations. The instantaneous power consumption of the load, denoted by $P_{\text{CL}}[k]$, is constrained by its rated capacity:
\begin{equation}
    0 \leq P_{\text{CL}}[k] \leq \overline{P}_{\text{CL}}\,,
\end{equation}
where $\overline{P}_{\text{CL}}$ represents the maximum allowable power consumption. For this study, it is assumed that the controllable load is fully dispatchable with negligible ramping limitations, i.e., it can increase or decrease its power consumption instantaneously within the feasible bounds. This assumption allows the load to actively participate in flexibility provision in conjunction with battery storage and generation assets. By integrating these models of generation, controllable load, and battery, the subsequent section formalizes the notion of flexibility and characterizes the capacity of the HES to provide maximum flexibility.

\section{Flexible HES dispatch}
\label{Flexibility}
One of the fundamental requirements for reliable power system operation is maintaining the balance between net demand and generation. Flexibility, in this context, refers to the ability of the system to absorb or deliver power to maintain grid stability by providing grid services such as frequency regulation~\cite{Mohandes2019}. Key attributes of flexibility include ramp rate, power capacity, and energy capacity, which together define how quickly and for how long a resource can respond~\cite{Makarov2009}. Power capacity defines the upper and lower limits of the power dispatchable to the grid, while energy capacity determines the duration for which a resource can sustain flexibility provision.

As grid operators seek to integrate more renewable and distributed resources, HES that combine generation, storage, and flexible loads offer new opportunities to deliver flexibility. Rather than managing each component independently, treating the HES as a unified resource with symmetric up/down regulation capability ($\pm C$ MW) simplifies grid interaction and bidding strategies. The key assumption here is that every asset can deliver any power within that region. However, determining the maximum symmetric regulation capacity $\pm \Delta P_{\text{hes}}$  that can be reliably offered each hour is challenging. The intermittency of solar generation introduces variability in available power, directly affecting system flexibility. This uncertainty complicates internal dispatch decisions and raises a key question: \textit{What is the largest symmetric flexibility the HES can commit while ensuring feasible operation across all components?} We address this question by formally defining and characterizing the symmetric HES capacity for flexibility, as elaborated subsequently.

For an HES, the net dispatch power $P_{\text{hes}}[k]$ can be represented as the sum of a nominal output $P_0[k]$, referred to as virtual generation, and a flexible component $\Delta P_{\text{hes}}[k]$ representing its ability to ramp its output up or down like a virtual battery. This decomposition yields:
\begin{equation}
P_{\text{hes}}[k] = P_0[k] + \Delta P_{\text{hes}}[k].
\label{HES power}
\end{equation}
Operationally, the HES maintains a nominal output at $P_0[k]$, while actively adjusting its power by $\pm \Delta P_{\text{hes}}[k]$ to follow the regulation signal. This flexible operating range is conceptually represented as a ``flexibility envelope"~\cite{Nosair2015}, within which the HES must be capable of delivering or absorbing power at any given time. In the considered HES, the PV output $P_{\text{PV}}[k]$ is treated as a non-dispatchable input, while the controllable load $P_{\text{CL}}[k]$ and battery $P_{\text{batt}}[k]$ are fully dispatchable within their rated limits $\overline{P}_{\text{CL}}$ and $\overline{P}_{\text{batt}}$, respectively. Unless curtailment is explicitly permitted, this reflects common operational practice in which solar generation is treated as non-dispatchable, while the battery and controllable load are actively scheduled. Under this formulation, the flexibility envelope is characterized in terms of the admissible deviation $\Delta P_{\text{hes}}[k]$ around $P_0[k]$. Having established the maximal realizable flexibility sets under different operating scenarios, the next step is to define a real-time allocation policy that maps a desired HES dispatch to feasible setpoints for individual assets.

\subsection{Real-Time Power Allocation Protocol}
\label{subsec:allocation_protocol}
This subsection presents a real-time power allocation protocol that disaggregates the net HES dispatch $P_{\text{hes}}[k]$ into component-level setpoints while respecting the flexibility envelopes. The objective of the allocation is to exactly realize any admissible $\Delta P_{\text{hes}}[k]$ using feasible combinations of controllable load and battery power, without violating component power or energy constraints. To achieve maximum symmetric flexibility in the absence of PV curtailment and without enforcing purely renewable-driven operation of the controllable load, the allocation strategy must prioritize how flexibility is shared between the controllable load and the battery. Giving priority to the controllable load is advantageous due to its intrinsic process value (e.g., hydrogen production), while limiting unnecessary high-frequency charge-discharge cycling of the battery to preserve its lifetime and sustain long-term regulation capability. Consequently, the power allocation strategy emphasizing controllable load utilization is given by,
\begin{subequations}
\begin{align}
    P_{\text{CL}}[k] &:= \biggl[P_{\text{PV}}[k] - P_0[k]-\Delta P_{\text{hes}}[k]\biggr]^{\overline{P}_{\text{CL}}}_0 \label{CL_load_disp}\\
    P_{\text{batt}}[k] &:= \biggl[P_0[k]+\Delta P_{\text{hes}}[k]- P_{\text{PV}}[k]+ P_{\text{CL}}[k]\biggr]^{\overline{P}_{\text{batt}}}_{\overline{P}_{\text{batt}}}\label{batt_disp}
\end{align}
\label{eq:power_alloc}
\end{subequations}
This allocation is valid provided that the controllable load is independent of the PV output, i.e., it is not required to be exclusively powered by the PV system.
\begin{proof}
Fix any $k \in \mathcal{K}$ and any feasible $\Delta P_{\text{hes}}[k]$ in flexibility set i.e. $\Delta P_{\text{hes}}[k] \in \mathcal{P}$ and consider the allocation rules \eqref{CL_load_disp}--\eqref{batt_disp}.\\
\textit{Power feasibility:}
By construction, both allocations use limits, which ensure
\[
0 \le P_{\text{CL}}[k] \le \overline P_{\text{CL}}, \qquad
-\overline P_{\text{batt}} \le P_{\text{batt}}[k] \le \overline P_{\text{batt}}.
\]
\textit{Power balance:}
Substituting \eqref{CL_load_disp}--\eqref{batt_disp} into the HES power balance \eqref{eq:hes_power_bal} gives
\[
P_{\text{PV}}[k] - P_{\text{CL}}[k] + P_{\text{batt}}[k]
= P_0[k] + \Delta P_{\text{hes}}[k],
\]
where the limits does not activate for
$\Delta P_{\text{hes}}[k] \in \mathcal{P}$ by definition of the flexibility set.\\
\textit{SoC feasibility:}
Since $|P_{\text{batt}}[k]| \le \overline P_{\text{batt}}$ and
$\text{SoC}[k]$ is assumed to satisfy operational bounds,
the battery state at $k+1$ remains feasible under the standard SoC dynamics.
Therefore, every $\Delta P_{\text{hes}}[k] \in \mathcal{P}$ is realizable,
which completes the proof.
\end{proof}
In this strategy, once the controllable load reaches its rated power limit, the remaining power is allocated to the battery. Further, under this power allocation rule, we can define the symmetric HES flexibility. Note that for different scenario as shown in Table~\ref{tab:comparison}, power allocation might follow different strategies.

\subsection{Scenario 1: symmetric HES flexibility with no curtailment}

\begin{claim} \label{claim_Max_HES}[Maximal Symmetric HES Flexibility without PV curtailment]
Consider an HES composed of a battery $P_{\text{batt}}[k]\in [-\overline P_{\text{batt}},\overline P_{\text{batt}}]$, a controllable load $P_{\text{CL}}[k]\in [0,\overline P_{\text{CL}}]$, and non-dispatchable PV generation $P_{\text{PV}}[k]$. 
If nominal power output $P_0[k] := P_{\text{PV}}[k]-\frac{1}{2}  \overline P_{\text{CL}}$, then the maximum symmetric range for flexibility of the HES is $\mathcal{\overline{P}}$, where 
\begin{align}
    \mathcal{\overline{P}} := \bigg[ -\overline P_{\text{batt}} - \frac{1}{2} \overline \Delta P_{\text{hes}}[k] \in P_{\text{CL}},
    \overline P_{\text{batt}} + \frac{1}{2} \overline P_{\text{CL}} \bigg].
\end{align}
\end{claim}

\begin{proof}
Proving the claim requires showing
\begin{enumerate}
    \item Any symmetric expansion of $\overline{\mathcal{P}}$ yields an unrealizable $\Delta P_{\text{hes}}$. 
    \item Every value in $\mathcal{\overline{P}}$ is realizable.
\end{enumerate}
Conversely, the realizability of every value in  $\mathcal{\overline{P}}$ (the second condition) is demonstrated through the proposed power allocation protocol. Specifically, the allocation scheme maps any $P[k] \in \overline{\mathcal{P}}$ to feasible battery and controllable load dispatch that satisfy all component-level constraints.


Extending $\mathcal{\overline{P}}$ symmetrically by any $\epsilon > 0$ results in unrealizable power values. Consider the extended range $\mathcal{P}_\epsilon$. For the upper bound, $\Delta P'_{\text{hes}}[k] = \overline P_{\text{batt}} + \frac{1}{2} \overline P_{\text{CL}} + \epsilon$ and corresponding $P'_{\text{hes}}[k] = \overline P_{\text{batt}} + P_{\text{PV}}[k] + \epsilon$. The maximum achievable power is $\overline P_{\text{hes}}[k] = \overline P_{\text{batt}} + P_{\text{PV}}[k]$. Since $\epsilon > 0$, $P'_{\text{hes}}[k] > \overline P_{\text{hes}}[k]$, and thus $P'_{\text{hes}}[k]$ is not realizable. Similarly, for the lower bound, let the extended range be $\Delta P''_{\text{hes}}[k] = -\overline P_{\text{batt}} - \frac{1}{2} \overline P_{\text{CL}} - \epsilon$. The corresponding HES power is $P''_{\text{hes}}[k] = P_{\text{PV}}[k] - \overline P_{\text{CL}} - \overline P_{\text{batt}} - \epsilon$. The minimum physical power limit is $\underline{P}_{\text{hes}}[k] = P_{\text{PV}}[k] - \overline P_{\text{CL}} - \overline P_{\text{batt}}$. Since $\epsilon > 0$, then $P''_{\text{hes}}[k] < \underline{P}_{\text{hes}}[k]$, confirming that $P''_{\text{hes}}[k]$ is also not realizable. Therefore, $\mathcal{\overline{P}}$ represents the maximal symmetric capacity for HES flexibility, as any expansion includes unrealizable power values. 
\end{proof}

\begin{remark}
Note that the above Claim~\ref{claim_Max_HES} assumes that PV generation is non-dispatchable, while the load is controllable. If controllability is instead assigned to the generation resource, the resulting $\overline{ \mathcal{P}}$ can be derived by judiciously designing $P_0[k]$. The $\Delta P_{\text{hes}}[k]$ in that case will determine by accurate forecast of PV generation which is beyond the scope of this work. Note that this claim is defined as case $S_{1}$ in Table~\ref{tab:comparison}.
\end{remark}

\subsection{Scenario 2: symmetric HES with generation consumed}

In previous case, the key point is that the controllable load operates at half of its rated capacity, regardless of the available PV generation; i.e., the load may rely on grid power, which is not necessarily clean. However, if the objective is to maximize symmetric flexibility while ensuring the controllable load is powered exclusively by renewable generation, the nominal power and the corresponding symmetric flexibility needs to be defined differently. In that case, the load is dispatched at half of the available PV, i.e., $P_{0}[k] =\frac{1}{2} \min \left( P_{\text{PV}}[k], \overline P_{\text{CL}}\right)$ and flexibility of the HES is 

\begin{align}
    \Delta P_{\text{hes}}[k] \in \bigg[ -\overline P_{\text{batt}} - \frac{1}{2} \min\{ P_{\text{PV}}[k], \overline P_{\text{CL}\}}, \nonumber \\
    \overline P_{\text{batt}} + \frac{1}{2} \min\{ P_{\text{PV}}[k], \overline P_{\text{CL}}\} \bigg] = \mathcal{\overline{P}_{\text{H}}}.
\end{align}

This adjustment ensures purely renewable-driven operation of the controllable load imposes stricter constraints on HES capacity. This case is defined as $S_{2}$ in Table~\ref{tab:comparison}. For $S_2$, the power allocation is defined through a linear relationship between $P_{\text{CL}}[k]$ and $P_{\text{hes}}[k]$, established using two operating points as shown in Fig.~\ref{Power_allocation}. This formulation assumes $\Delta P_{\text{hes}} \in \mathcal{P}_{\text{H}}$, $P_0[k] = \frac{1}{2}P_{\text{PV}}[k]$, and $\overline{P}_{\text{CL}} \geq P_{\text{PV}}[k] \geq 0 \ \forall~k$. Specifically, the maximum upward capacity corresponds to $P_{\text{PV}}[k] + \overline{P}_{\text{batt}}$, achieved when the controllable load is off (i.e., $P_{\text{CL}}[k] = 0$), while the maximum downward capacity is $-\overline{P}_{\text{batt}}$, achieved when the load consumes all available PV power (i.e., $P_{\text{CL}}[k] = P_{\text{PV}}[k]$).
\begin{figure}[!htb]
\centering
\includegraphics[width=0.45\textwidth]{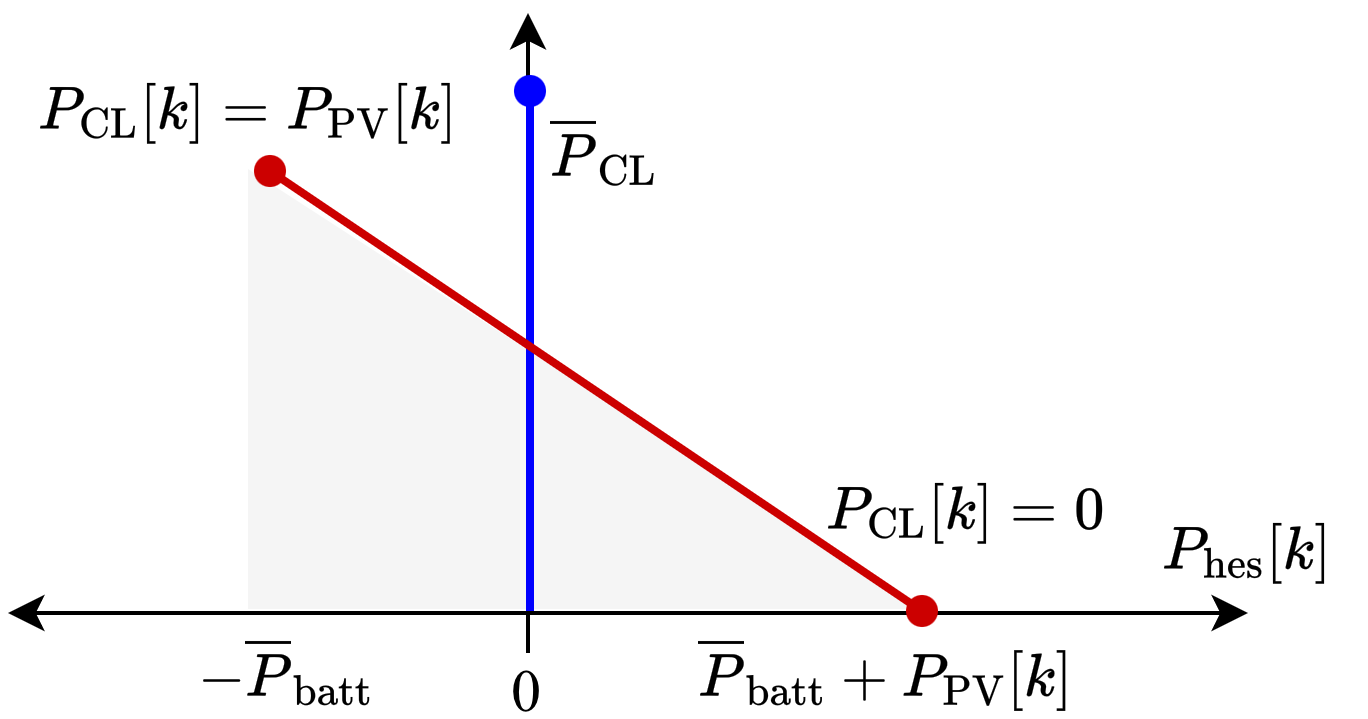}
\caption{Controllable load power allocation for green hydrogen production, assuming $\overline{P}_{\text{CL}}\ge P_{\text{PV}}[k]$ with the maximum flexibility of $\pm \overline{P}_{\text{batt}}\pm P_{\text{PV}}[k]$ and $P_0[k] = \frac{1}{2}P_{\text{PV}}[k]$.}
\label{Power_allocation}
\centering
\end{figure}
Based on this linear relationship, following power allocation protocols are defined for $P_{\text{CL}}[k]$ and $P_{\text{batt}}[k]$:
\begin{subequations}
\begin{align}
    P_{\text{CL}}[k] &:= \biggl[P_{\text{PV}}[k]\left(\frac{\overline P_{\text{batt}} + \frac{1}{2}P_{\text{PV}}[k]-\Delta P_{\text{hes}}[k]}{2\overline P_{\text{batt}} + P_{\text{PV}}[k]} \right)\biggr]^{\overline{P}_{\text{CL}}}_0 \label{CL_power} \\
    P_{\text{batt}}[k] &:= \biggl[\overline P_{\text{batt}}\left(\frac{ 2\Delta P_{\text{hes}}[k]}{2\overline P_{\text{batt}} + P_{\text{PV}}[k]}\right)\biggr]^{\overline{P}_{\text{batt}}}_{-\overline{P}_{\text{batt}}} \label{batt_power}
\end{align}
\label{eq:power_alloc_2}
\end{subequations}
This allocation also follows the allocation rules defined by \eqref{eq:power_alloc}. 

\subsection{Scenario 3: HES demand matches supply}
\begin{remark}
On the other hand, if general demand process output (e.g. hydrogen production from electrolyzer) is very valuable, we can maximize this output by matching $P_{\text{CL}}[k] = P_{\text{PV}}[k]$, which reduces flexibility available to HES. The new nominal HES output is $P_0[k] = 0$ around which the flexibility range is decreased to
$\mathcal{P}_{\text{min}} := \left[ -\overline P_{\text{batt}}, \overline P_{\text{batt}}\right],
$
which means that the HES just represents a traditional battery from the perspective of grid operators as defined in case $S_{3}$ in Table~\ref{tab:comparison}.
\end{remark}
 
\begin{table*}[!htb]
\centering
\renewcommand{\arraystretch}{1.3}
\caption{Nominal Power and Flexibility Bounds of a Hybrid Energy System Under Different PV Scenarios}
\label{tab:comparison}
\begin{tabular}{|c|c|c|c|c|}
\hline
\thead{Scenario} & \thead{$P_0[k]$} & \thead{$\Delta P_{\text{hes}}$} & \thead{Dispatch Implications} & \thead{Response Profile} \\
\hline
\multirow{3}{*}{No PV Curtailment} 
& $P_{\text{PV}}[k] - \frac{1}{2}\overline{P}_{\text{CL}}$ 
& $\pm \overline{P}_{\text{batt}} \pm \frac{1}{2}\overline{P}_{\text{CL}}$ 
& $S_1$=: Max Flex – No PV Curtailment 
& Symmetric \\
& $\frac{1}{2} \min(P_{\text{PV}}[k], \overline{P}_{\text{CL}})$ 
& $\pm \overline{P}_{\text{batt}} \pm \frac{1}{2} \min(P_{\text{PV}}[k], \overline{P}_{\text{CL}})$ 
& $S_2$: Flex Sustainable Load 
& Symmetric \\
& $0$ 
& $\pm \overline{P}_{\text{batt}}$ 
& $S_3$: Min Flex / Max Sustainable Load 
& Symmetric \\
\hline
\multirow{2}{*}{PV Curtailment Allowed} 
& $\frac{1}{2} \big(P_{\text{PV}}[k] - \overline{P}_{\text{CL}}\big)$ 
& $\pm \overline{P}_{\text{batt}} \pm \frac{1}{2} \big(P_{\text{PV}}[k] + \overline{P}_{\text{CL}}\big)$ 
& $S_4$: Max Symmetric Flexibility 
& Symmetric \\
& $0$ 
& $\big[-\overline{P}_{\text{batt}} - \overline{P}_{\text{CL}},~\overline{P}_{\text{batt}} + P_{\text{PV}}[k]\big]$ 
& $S_5$: Max Flex / Max PV Curtailment 
& Asymmetric \\
\hline
\end{tabular}
\end{table*}

\subsection{Scenario 4: symmetric HES flexibility with curtailment}

\begin{claim}
Clearly, Under $\overline{\mathcal{P}}$, $\overline{\mathcal{P}_{\text{H}}}$ and $\mathcal{P}_{\text{min}}$, solar PV generation is not curtailed; since controllable load absorbs solar PV or it is injected into grid. If we permit curtailment (i.e., if flexibility is far more valuable than solar PV), then we can increase symmetric capacity with $P_0[k] = \frac{1}{2}(P_{\text{PV}}[k]-\overline{P}_{\text{CL}})$ and flexibility range increased to
$$\mathcal{P}_{\text{curt}} := \left[ -\overline P_{\text{batt}} - \frac{1}{2}(P_{\text{PV}}[k]+\overline{P}_{\text{CL}}), \overline P_{\text{batt}} + \frac{1}{2}(P_{\text{PV}}[k]+\overline{P}_{\text{CL}})\right]
$$

Here, the factor $\frac{1}{2}$ arises directly from allowing partial solar PV curtailment. 
\end{claim}
\begin{proof} 
Let $\Delta P_{\text{hes}}[k]$ be extended by $\epsilon > 0$. For the upper bound, let $\Delta P'_{\text{hes}}[k] = \overline P_{\text{batt}} + \frac{1}{2}(P_{\text{PV}}[k] + \overline{P}_{\text{CL}}) + \epsilon$. The resulting power setpoint is $P'_{\text{hes}}[k] = P_0[k] + \Delta P'_{\text{hes}}[k] = P_{\text{PV}}[k] + \overline P_{\text{batt}} + \epsilon$. This exceeds the operational limits of the generation and battery, rendering it unrealizable. By symmetry, an extension to the lower bound exceeds the combined charging and load capacity. Thus, $\mathcal{P}_{\text{curt}}$ is maximal. This case is defined as $S_{4}$ in Table~\ref{tab:comparison}
\end{proof}

\subsection{Scenario 5: asymmetric HES flexibility with curtailment}
\begin{claim}
If symmetry is not the major constraint, flexibility can further be increased with $P_0[k] =0$, allowing the system to utilize full resource capacity i.e.
$\overline{\mathcal{P}}_{\text{curt}} := \left[ -\overline P_{\text{batt}} - \overline{P}_{\text{CL}}, \overline P_{\text{batt}} + P_{\text{PV}}[k]\right]
$
This represents the absolute maximum achievable flexibility when symmetry is not enforced.
\end{claim}

\begin{proof} 
With $P_0[k]=0$, any flexibility request $\Delta P_{\text{hes}}[k]$ outside $\overline{\mathcal{P}}_{\text{curt}}$ by $\epsilon > 0$ maps directly to a power setpoint $P_{\text{hes}}[k]$ that exceeds either the maximum output ($\overline P_{\text{batt}} + P_{\text{PV}}[k]$) or maximum consumption ($-\overline P_{\text{batt}} - \overline P_{\text{CL}}$). Thus, $\overline{\mathcal{P}}_{\text{curt}}$ represents the absolute maximum flexibility of the system. This case is defined as $S_{5}$ in Table~\ref{tab:comparison}
\end{proof}

Table~\ref{tab:comparison} summarizes the resulting flexibility envelopes and respective nominal power levels for all these scenarios. It can be concluded that permitting PV curtailment increases the achievable flexibility envelope. However, this increased flexibility comes at the expense of reduced PV utilization, potentially lowering system efficiency. Therefore, while PV curtailment can enhance dispatchability, it may not be the preferred strategy in scenarios where maximizing renewable energy is prioritized. For any $\Delta P_{\text{hes}}[k] \in \mathcal{P}$, we can directly allocate controllable load and battery to achieve the desired HES output based on different objective using \eqref{eq:power_alloc} or \eqref{eq:power_alloc_2}. Certainly, these allocations are not unique but are sufficient. The defined allocation protocols are then used to dispatch the HES power based on the capacity bid in the market. This $\pm$MW capacity bid for market participation is discussed in the next section.

\section{Frequency regulation capacity of HES}
\label{sec:market_setting}
The frequency regulation market generally operates on a pay-for-performance model where participants (e.g., HES owners) submit a bid for regulation capacity $C$, in MW~\cite{Bolun2018}. The payment is then determined by this bid $C$, a performance score $x_{\text{p}}$, mileage $M_{\text{d}}$ of the regulation signal. The payment, $\lambda_{R}$ is typically calculated as~\cite{Bolun2016}:
\begin{equation}
  \lambda_{\text{R}}= x_{\text{p}}C(\lambda_{\text{C}}+M_{\text{d}}\lambda_{\text{M}}) 
  \label{capacity_credit}
\end{equation}
where $\lambda_{\text{C}}$ and $\lambda_{\text{M}}$  are the clearing prices for the market capacity and mileage, respectively. The term mileage $M_{\mathrm{d}}$, defined as the total absolute movement of the regulation signal \textbf{r} for a given interval~\cite{pjm2025manual12}:
\begin{equation}
    M_{\mathrm{d}} = \sum_{k=0}^{N-1} 
    \left| r[k+1] - r[k] \right|.
\end{equation}
Mileage represents the required movement of the resource and typically influences compensation in markets that explicitly reward ramping capability. In HES case, this bid represents the maximum power from the system’s nominal power $P_0$ that the HES can reliably track. The corresponding HES power at time $k$ is therefore given by:
\begin{equation}
P_{\text{hes}}[k] = P_0[k] + C r[k], \quad \forall k \in \mathcal{K}
\end{equation}

where $C\mathbf{r}$ denotes the awarded regulation signal, with $C$ representing the cleared regulation capacity (MW) and $r[k]$ the normalized regulation signal at $k \in \mathcal{K} = \{0,\ldots,N-1\}$. This implies that the HES is required to modulate its power output within $\pm C$ MW in response to the regulation command. While a `conservative' underbidding of capacity may lead to lost revenue opportunities, `aggressive' overbidding can result in tracking errors, poor performance scores, and financial penalties. Therefore, accurate bid sizing is essential for both compliance and profitability. One bidding approach ensuring maximum HES flexibility for the given hour in the market is given by,
\begin{equation}
C = \frac{\norm{\Delta \textbf{P}_{\text{hes}}}_\infty}{\norm{\textbf{r}}_{\infty}}
\end{equation}
In this straightforward case, where the HES's flexible operating range $\Delta P_{\text{hes}}$ does not depend on variable solar PV output (e.g., as in the first row of Table \ref{tab:comparison}, where $\Delta P_{\text{hes}}$ = $\overline{P}_{\text{batt}}+ \frac{1}{2}\overline{P}_{\text{CL}}$), the C is readily determined as $\overline{P}_{\text{batt}}+ \frac{1}{2}\overline{P}_{\text{CL}}$. However, in cases where the available HES flexibility $\Delta P_{\text{hes}}$ depends explicitly on the PV output $P_{\text{PV}}[k]$, capacity bidding becomes more challenging due to the strong temporal and seasonal variability of solar generation, as shown in Fig.~\ref{fig:PV variability}. Since regulation capacity bids are submitted in the hour-ahead market, this variability introduces uncertainty in the realizable flexibility during real-time operation. A detailed optimal capacity bidding strategy under uncertainty, formulated using chance constraints, is presented by the authors in \cite{mishra2025}. To ensure reliable participation and maintain quality of service under such conditions, system operators rely on performance-based scoring metrics, as discussed in Section~\ref{PJM_market} for the PJM frequency regulation market, which remains one of the highest-priced regulation markets in the United States~\cite{ANL_survey_2016}.

\subsection{Optimal HES Dispatch}
\label{subsec:optimization_framework}
To benchmark the performance of the proposed rule-based real-time power allocation protocol and to establish an upper bound on the achievable regulation performance, we formulated an optimization-based dispatch policy for HES participation in the frequency regulation market. The problem is posed as a mixed-integer program (MIP) over a discrete time horizon indexed by $k \in \mathcal{K} = \{0,\ldots,N-1\}$. Among the HES configurations summarized in Table~\ref{tab:comparison}, the optimization framework is developed for the case that maximizes flexibility without curtailing PV generation (row~1), as this scenario is the most relevant for practical market participation. The same formulation can be readily extended to other configurations with minor modifications to the asset constraints. The complete mathematical formulation of the optimization problem is presented below. 

The objective is to minimize the tracking error between the HES flexible power response and the awarded regulation signal, expressed as $C\mathbf{r}$. This formulation assumes perfect knowledge of the regulation signal, perfect solar forecast and resource availability over the optimization horizon and therefore serves as a performance benchmark rather than a real-time implementable controller.

\textbf{Objective Function:}
\begin{equation}
\min_{\mathbf{y}}~~\sum_{k=0}^{N-1} \lvert Cr[k] - \underbrace{P_{\text{hes}}[k]+ P_{0}[k]}_{\text{flexible component }}\rvert\,,  \\
\label{Objective}
\end{equation}

\textbf{Subject to:}
\begin{subequations}
\begin{align}
    P_{\text{hes}}[k] -P_{\text{batt}}[k] - P_{\text{PV}}[k] + P_{\text{CL}}[k] = 0 &\\
    P_{0}[k]-P_{\text{PV}}[k]+\frac{1}{2}\overline{P}_{\text{CL}} = 0&\\
    P_{\text{batt}}[k] - P^{\text{c}}_{\text{batt}}[k] - P^{\text{d}}_{\text{batt}}[k] = 0&\\
    E[k+1] - E[k] + \frac{\Delta t}{E_{\text{C}}} \Big( \eta P^{\text{c}}_{\text{batt}}[k] 
    + \frac{1}{\eta} P^{\text{d}}_{\text{batt}}[k] \Big) = 0 &\\
    E[0]=E_0 &\\
    P^{\text{d}}_{\text{batt}}[k] - Z[k] \overline{P}_{\text{batt}} \leq 0& \label{batt_bounds}\\
    -P^{\text{c}}_{\text{batt}}[k] - (1-Z[k])  \overline{P}_{\text{batt}} \leq 0& \\
    P^{\text{c}}_{\text{batt}}[k] \leq 0& \label{Ch_bounds} \\
    0 \leq P^{\text{d}}_{\text{batt}}[k] & \label{Dis_bounds} \\
    \underline{E} \leq E[k+1] \leq \overline{E}  \label{SoC_bounds} \\
    Z[k] \in \{0, 1\}&\\
     0 \leq P_{\text{CL}}[k] \leq \overline{P}_{\text{CL}}& \label{CL_bounds}
\end{align}
\label{Constraints}
\end{subequations}

where the decision variables $\mathbf{y}$ consist of $P_{\text{CL}}[k]$, $P_{\text{batt}}^{\text{d}}[k]$, $P_{\text{batt}}^{\text{c}}[k]$, and $Z[k]$, for all times $k\in\mathcal{K}$. The problem is subject to a set of operational and physical constraints given by \eqref{Constraints}. These include: the overall HES power balance, battery charge ($P^{\text{c}}_{\text{batt}}[k]$) and discharge ($P^{\text{d}}_{\text{batt}}[k]$) components, enforced by the binary variable $Z[k]$ to prevent simultaneous charging and discharging~\cite{Mazen2023}. The inequality constraints define the bounds on maximum power and SoC respectively through \eqref{batt_bounds}-\eqref{CL_bounds}. Although derived for a closely related optimization problem, the KKT conditions in \cite{mishra2025} reveal structural properties of the optimal solution that align with the proposed real-time, rule-based dispatch. This correspondence shows that, under non-binding SoC constraints, the online power allocation achieves the same performance as the optimal policy.
\vspace{-0.25cm}
\section{Case Study: HES Participation in the Frequency Regulation Market}
\label{sec:PJM_case_study}
To validate the proposed maximum flexibility characterization and the associated real-time power allocation protocol, we have done the simulation studies using one year of PJM frequency regulation data sampled at 2~$s$ resolution. To realistically represent solar generation, high-resolution site-measured irradiance data are processed through an equivalent-circuit PV cell model~\eqref{PV_Cell_Power}, and the resulting output is scaled to MW-level PV capacity. The resulting PV profile captures both diurnal and seasonal variability, which directly impacts the available flexibility of the HES.

The battery and controllable load models explicitly account for power limits, energy capacity, and inverter efficiencies, as summarized in Table~\ref{tab:HES_params}. Battery SoC dynamics are modeled using charging and discharging efficiencies, and operational constraints are enforced at each dispatch interval $k \in \mathcal{K}$ using \eqref{battery model}. All HES assets are assumed to exhibit sufficiently fast response relative to the 2-second regulation signal, such that communication and actuation delays are negligible. This assumption is consistent with PJM requirements for fast-responding regulation resources and allows the analysis to focus on power allocation and SoC management rather than device-level dynamics. The relevant PJM frequency regulation market settings and performance metrics are detailed in the following subsection.

\begin{table}[!t]
\centering
\caption{HES Parameters}
\begin{tabular}{l c}
\hline
\textbf{Component} & \textbf{Value} \\
\hline
PV rated power $\overline{P}_{\!\text{PV}}$ & 3 MW \\
Controllable load rated power $\overline{P}_{\text{CL}}$ & 3 MW \\
Battery rated power $\overline{P}_{\text{batt}}$& 5 MW \\
Battery energy capacity & 5 MWh \\
Battery round-trip efficiency $\eta_{\text{d}} =\eta_{\text{c}} $ & 95\% \\
Initial SoC  $E_{0}$ & 0.50 p.u.\\
Max SoC  $\overline{E}$ & 0.90 p.u.\\
Min SoC  $\underline{E}$ & 0.10 p.u.\\
\hline
\end{tabular}
\label{tab:HES_params}
\end{table}
\subsection{Frequency Regulation Market setting}
\label{PJM_market}

The Pennsylvania-New Jersey-Maryland Interconnection (PJM) is a Regional Transmission Organization (RTO) that coordinates the movement of wholesale electricity across all or parts of thirteen states and the District of Columbia. PJM operates one of the largest competitive wholesale electricity markets in the world and is responsible for procuring and dispatching frequency regulation service to maintain system balance. Resources participating in PJM's regulation market are evaluated using performance score~\cite{pjm2025manual12}.

PJM previously offered two regulation products, slow Reg-A and fast Reg-D; however, recent market reforms have consolidated these into a single frequency regulation product with a unified, fast-responding control signal. Under the updated framework, resource performance is primarily evaluated using a precision-based metric that quantifies how accurately a resource follows PJM’s real-time regulation signal. Due to the limited availability of post-reform regulation data, this study uses the historical Reg-D signal as a representative fast regulation signal \textbf{r}, which is inherently bounded in the range of $\pm 1$~MW; this signal is treated as approximately energy neutral over short intervals and reflects the system operator’s real-time balancing requirement. The performance score is calculated as
\begin{equation}
    x_{\text{p}} = 1 - \frac{\lVert C\textbf{r} - \textbf{P}_{\text{hes}} \rVert_1}{C\lVert \textbf{r} \rVert_1}
    \label{xp}
\end{equation}
The score represents the normalized tracking accuracy, with $x_{\text{p}} = 1$ indicating perfect alignment between the commanded and delivered power. PJM applies performance thresholds to ensure dependable regulation service. 
A performance score of at least $75\%$ is typically required to pass a regulation qualification test. Hence, we do not award any regulation revenue when the $x_{\text{p}}$ falls below 0.75. In our study, we assumed fast-responding assets with negligible delay and no ramp-rate limitations. Hence, the revenue becomes effectively proportional to the capacity bid $C$ and the achieved performance score $x_{\text{p}}$. Therefore, mileage $M_{\text{d}}$ is not used as a decision factor in this analysis.

\vspace{-0.25cm}
\subsection{Flexibility assessment for different scenarios}
Figs.~\ref{fig:P0_minipage}-\ref{fig:Delta_P_minipage} show the nominal operating points and corresponding flexibility tubes, defined by $\pm \Delta P_{\text{hes}}$, for the different HES scenarios $S_{\text{1}}$ to $S_{\text{5}}$ summarized in Table~\ref{tab:comparison}. Each scenario reflects a distinct nominal power $P_{0}$ and flexible $\pm P_{\text{hes}}$ among PV generation, battery storage, and controllable load, resulting in varying degrees of upward and downward regulation capability. The scenario~$S_{\text{3}}$ exhibits the lowest flexibility, as the controllable load absorbs the entire PV output, effectively reducing the HES to a battery only resource with limited up and down flexibility. In contrast, Scenario~$S_{\text{5}}$ achieves the highest flexibility by allowing PV curtailment, thereby enabling the PV subsystem to act as a fully controllable generation resource. This results in an asymmetric flexibility envelope, making scenario ~$S_{\text{5}}$ particularly suitable for market frameworks where upward and downward regulation capacities can be offered independently. Both these scenarios have zero nominal power $P_{0}$. However, when PV curtailment is not permitted, scenario~$S_{\text{1}}$ provides the maximum achievable flexibility by optimally sharing power between the PV, controllable load and battery. Fig.~\ref{fig:S1_minipage} shows the flexibility tube for scenario~~$S_{\text{1}}$, where the nominal operating point is defined as $P_{0}[k] = P_{\text{PV}}[k] - \frac{1}{2}\overline{P}_{\text{CL}}$, yielding symmetric headroom for both upward and downward regulation.

\begin{figure}[!htb]
\centering
\subfloat[Nominal Power $P_{0}$]
{%
\includegraphics[width=0.48\textwidth]{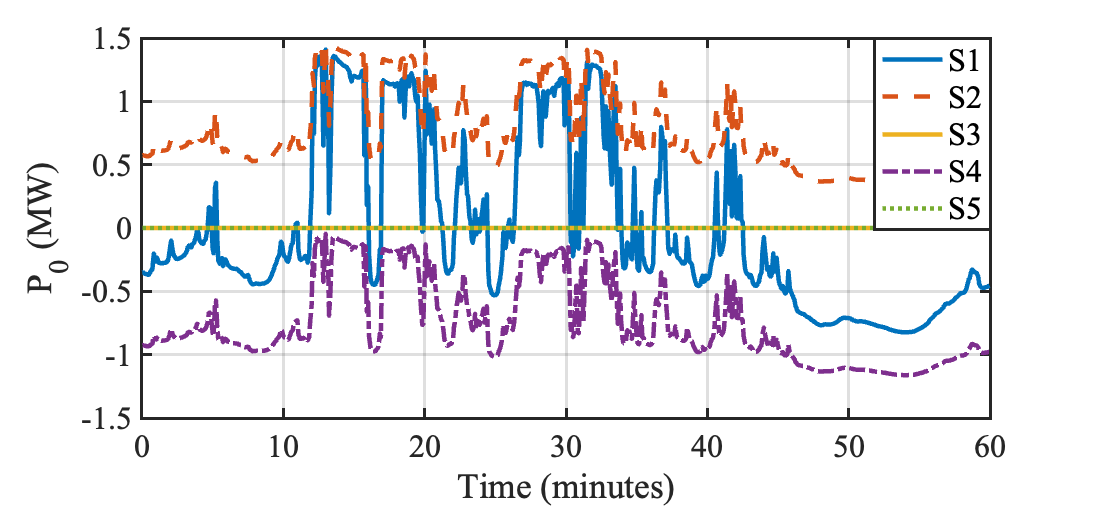}
    \label{fig:P0_minipage}
}
\hfill
\subfloat[Flexibility tubes for $\Delta P_{\text{hes}}$]{%
    \includegraphics[width=0.48\textwidth]{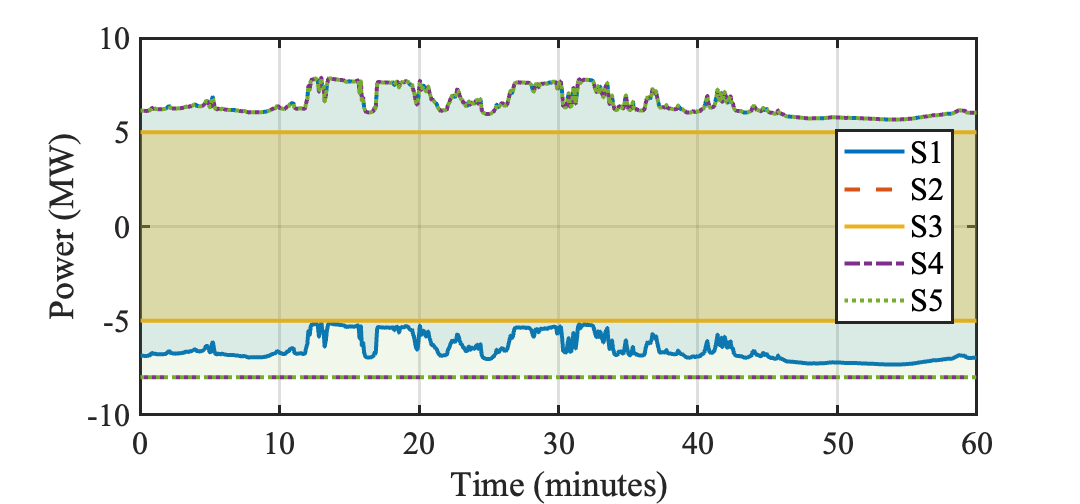}
    \label{fig:Delta_P_minipage}
}
\vspace{1em}
\subfloat[Illustrating HES capacity for scenario $S_1$: $P_{0} \pm \Delta P_{\text{hes}}$]{%
\includegraphics[width=0.48\textwidth]{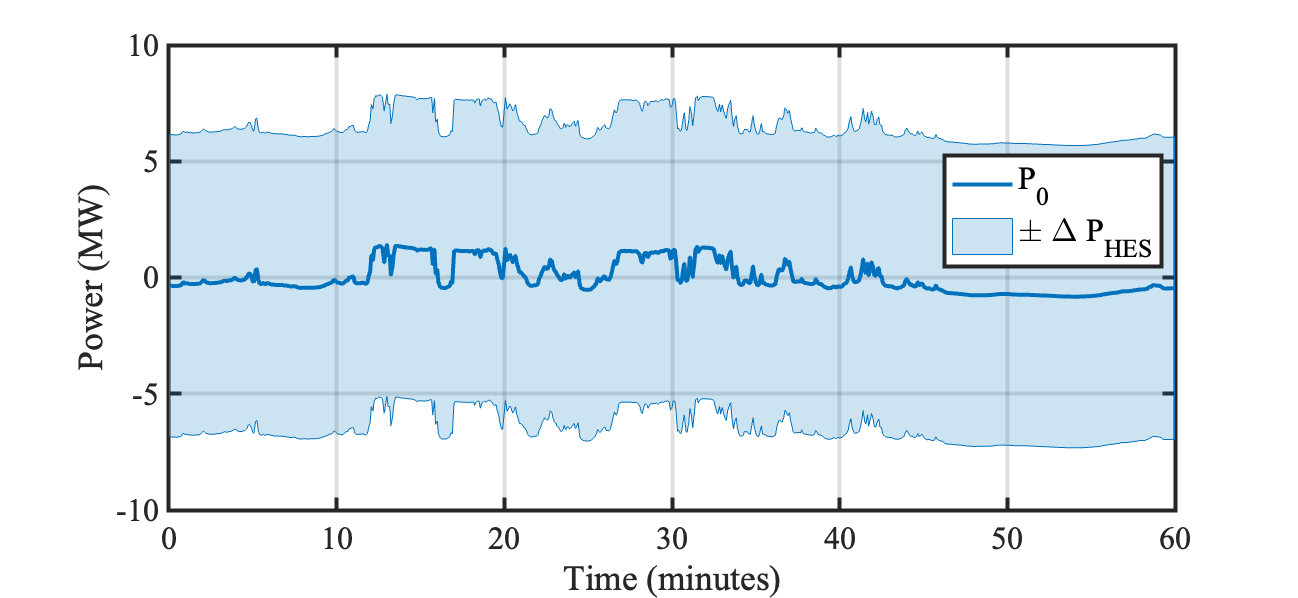}
    \label{fig:S1_minipage}
}
\caption{Flexibility tube defined around the nominal power trajectory $P_{0}[k]$ using the bounds $\pm \Delta P_{\text{hes}}[k]$ over a one-hour horizon for different scenarios as given in Table \ref{tab:comparison}. The HES is expected to track any power trajectory within this tube by ramping up or down from $P_{0}[k]$. In (c), $P_{0}[k] =P_{\text{PV}}[k] - \frac{1}{2}\overline{P}_{\text{CL}}$ for scenario $S_1$.}
\label{fig:Flexibility tube}
\end{figure}

\subsection{Real-time power allocation vs optimal dispatch for maximum flexibility}
To benchmark the real-time power allocation protocol (Section~\ref{subsec:allocation_protocol}), we compared it against the offline optimization framework (Section~\ref{subsec:optimization_framework}) using a bid capacity of $C = 6.5$~MW based on the maximum $\pm \Delta P_{\text{hes}}$. The MILP optimization problem is solved using \textit{Gurobi} solver.
Fig.~\ref{fig:HES_power} shows that both strategies track the regulation signal ``REG~SIGNAL" C\textbf{r} with identical precision, achieving a performance score $x_{\text{p}}$~=~1.0. 
Figs.~\ref{fig:Controllable load_power} and~\ref{fig:Battery_dispatch} compare the allocation of power among HES components between the optimal and real-time protocol, and Fig.~\ref{fig:SoC_opti} shows the battery SoC for that interval. While there are non-unique optimal allocations between PV power and controllable load consumption, the combined contribution $P_{\text{PV}} + P_{\text{H}_2}$ remains the same for both policies. It can be observed from this comparative study that the proposed real-time strategy achieves the same performance as the offline optimal dispatch strategy \eqref{Objective}-\eqref{Constraints} while remaining implementable in real time.

\begin{figure}[!htb]
\centering
\subfloat[Regulation signal and HES power]{%
    \includegraphics[width=0.48\textwidth]{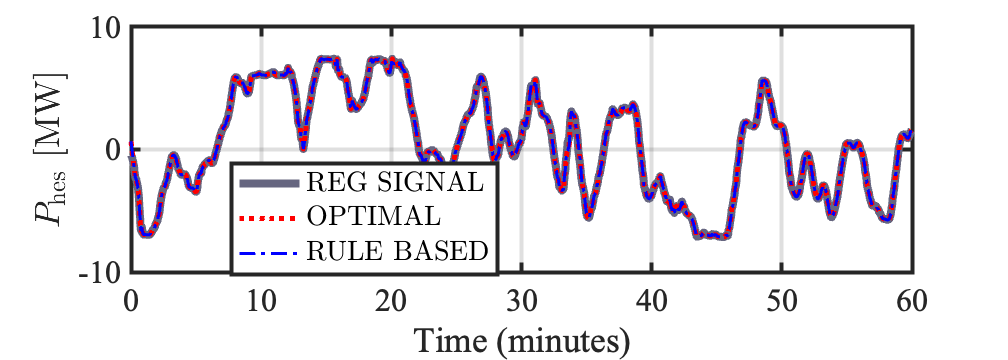}
    \label{fig:HES_power}
}
\hfill
\subfloat[Net controllable load and PV power: $P_{\text{pv}}$+$P_{\text{H}_2}$]{%
    \includegraphics[width=0.48\textwidth]{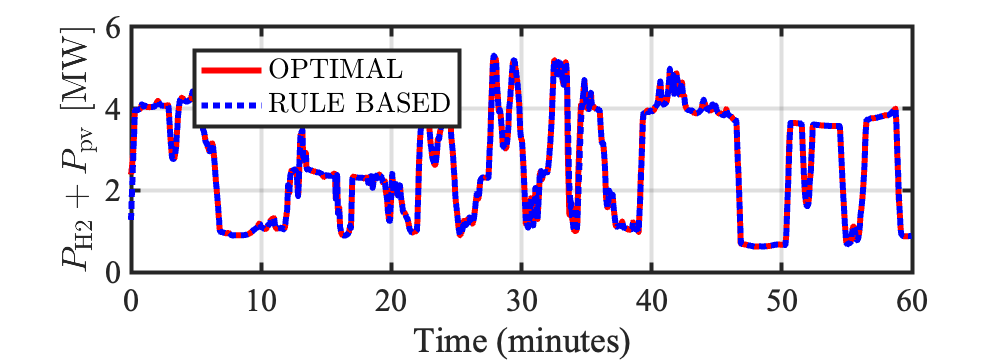}
    \label{fig:Controllable load_power}
}
\vspace{1em}
\subfloat[Battery net power dispatch $P_{\text{batt}}$]{%
    \includegraphics[width=0.48\textwidth]{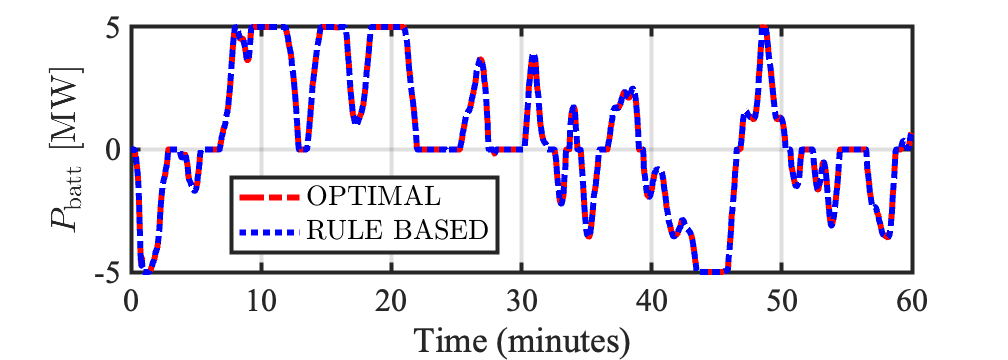}
    \label{fig:Battery_dispatch}
}
\hfill
\subfloat[Battery SoC $E$]{%
    \includegraphics[width=0.48\textwidth]{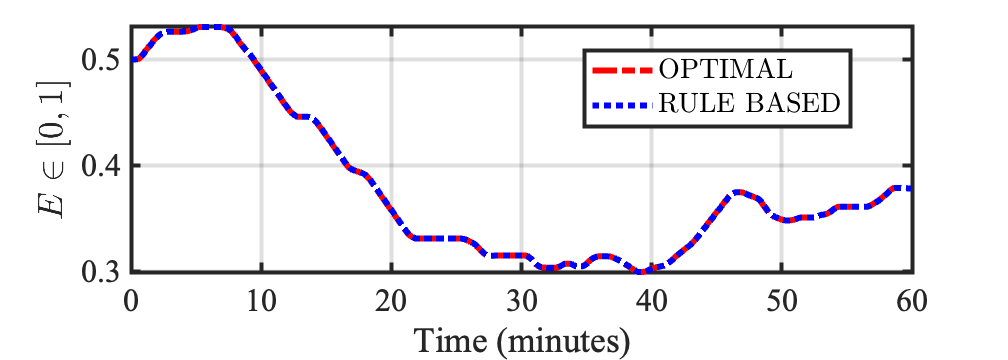}
    \label{fig:SoC_opti}
}
\caption{Comparison of real-time power allocation rule based dispatch and optimal offline dispatch for one representative hour with $C=\Delta P_{\text{hes}}$ = 6.5 MW. The plots show the net regulation signal $C\mathbf{r}$, $P_{\text{hes}}$, $P_{\text{gen}}$+ $P_{\text{CL}}$, $P_{\text{batt}}$,and battery SoC $E$. The performance $x_{\text{p}}=1$ for both strategies.}
\label{Opti_Flex}
\end{figure}

\subsection{Capacity bidding based on season variation in PV power}
Previously, maximum flexibility was achieved at the cost of continuous load operation, often requiring grid imports when PV power was unavailable. However, if purely renewable-driven operation is prioritized, such as scenario~$S_{2}$ offers the maximal sustainable flexibility. In this context, the capacity bid $C$ should also incorporate the multi-timescale (seasonal and intra-day) solar variability as shown in Fig.~\ref{fig:PV variability}. To address this, the bid can be decomposed into fixed and variable components: $C~=C_{0}+\hat{C}_{\text{PV}}$. Here, $C_{0}$ represents the power capacity based on the battery rating, while $\hat{C}_{\text{PV}}$ is determined by the solar forecast for the upcoming frequency regulation interval.

\begin{figure}[!htb]
\centering
\begin{minipage}{0.5\textwidth}
\centering
\includegraphics[width=\textwidth]{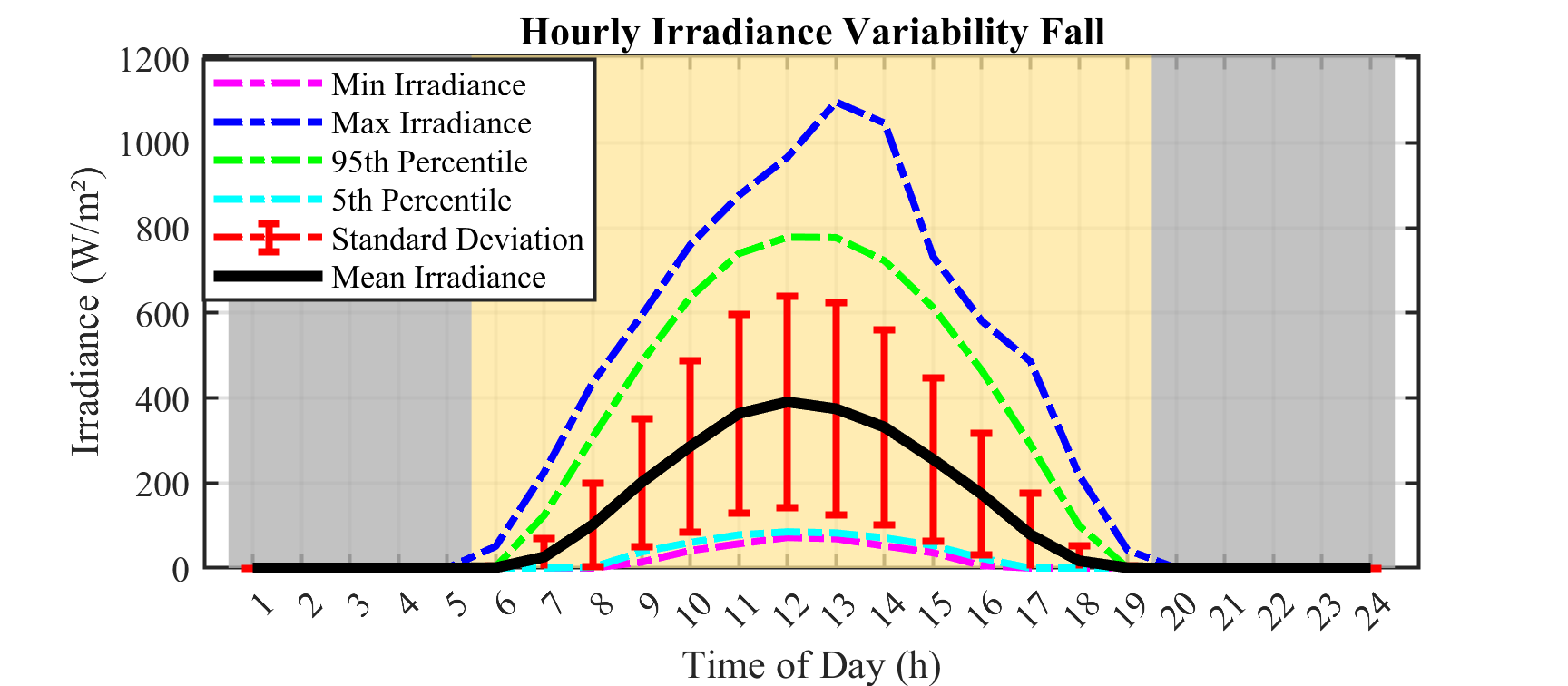}
\label{}
\end{minipage}
\hfill 
\begin{minipage}{0.5\textwidth}
\centering
\includegraphics[width=\textwidth]{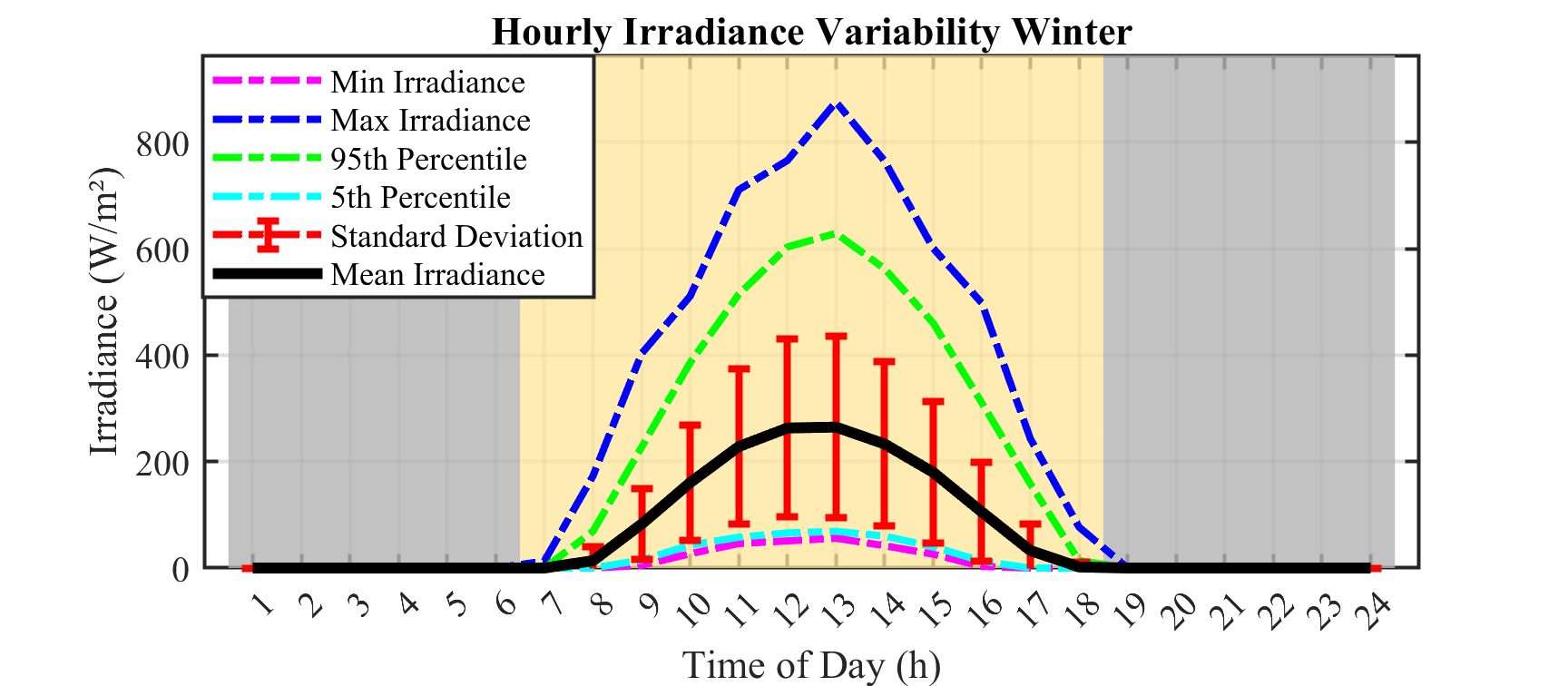}
\label{}
\end{minipage}

\vspace{1em} 
\begin{minipage}{0.5\textwidth}
\centering
\includegraphics[width=\textwidth]{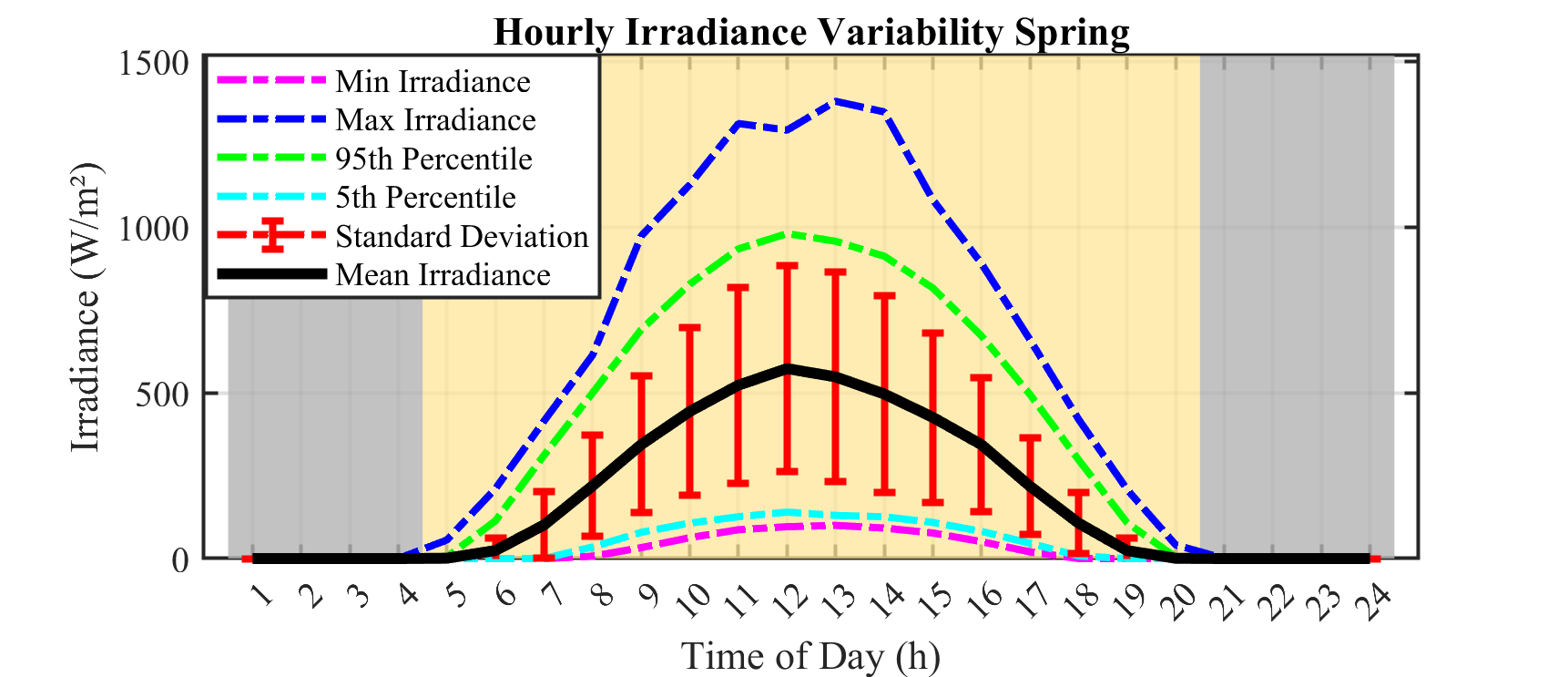}
\label{}
\end{minipage}
\hfill
\begin{minipage}{0.5\textwidth}
\centering
\includegraphics[width=\textwidth]{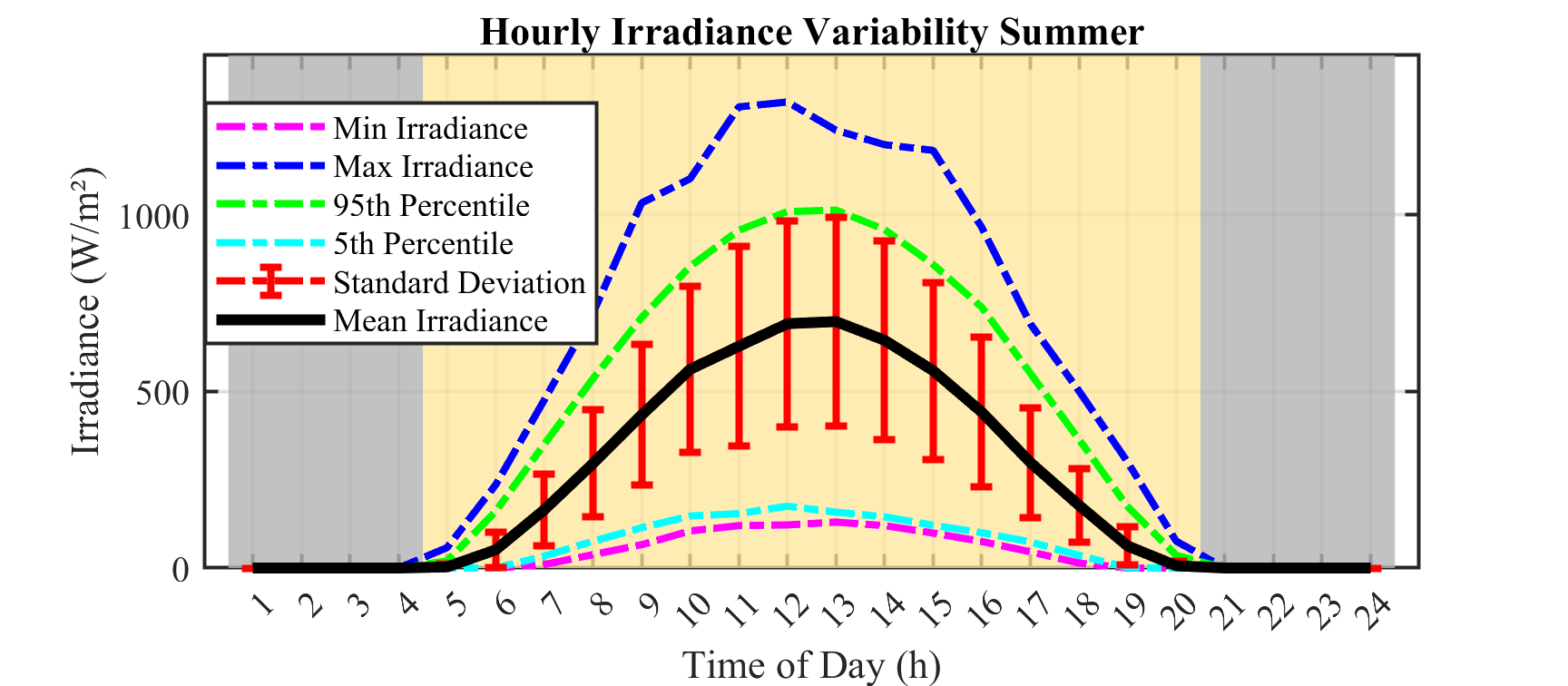}
\label{Solar variability}
\end{minipage}

\caption{Hourly Solar Irradiance Profiles: seasonal variation in mean, percentiles, extremes, and standard deviation for fall, winter, spring and summer for one year.}
\label{fig:PV variability}
\end{figure}

\begin{figure}[!htb]
\centering
\includegraphics[width=0.48\textwidth]{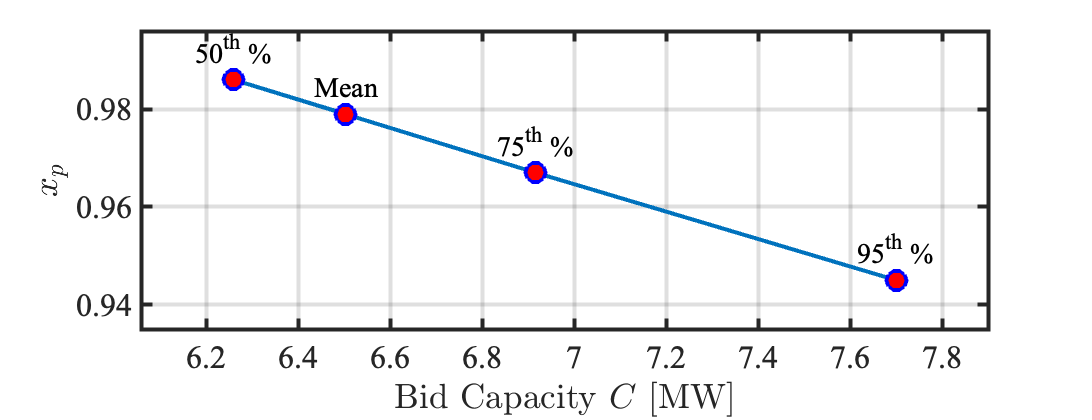}
\caption{Impact of PV statistic selection on HES Bid Capacity (C) and $x_{\text{p}}$. Markers indicate the PV power percentile (Mean, 50th, 75th, 95th) used to size the regulation bid.}
\label{fig:PV_bid}
\centering
\end{figure}

Implementing this decomposed bidding strategy requires selecting an appropriate statistical metric for $\hat{C}_{\text{PV}}$. Fig.~\ref{fig:PV_bid} shows the sensitivity of system performance to different statistical values of $P_{\text{PV}}$, explicitly highlighting the trade-off between maximizing bidding capacity and maintaining the performance. In this case, where PV rating is approximately $66.67\%$ of the battery size, the battery provides a substantial buffer against solar intermittency. However, if the PV capacity were to significantly exceed the battery rating, solar variability would become the dominant constraint, severely limiting the feasible regulation capacity. Consequently, a robust bidding strategy must dynamically adjust $\hat{C}_{\text{PV}}$ to align the regulation commitment with the available battery storage.

\subsection{State of charge drift correction for multi-hour participation}
\label{SoC_drift_correction}
The allocation protocol proposed in Section~\ref{subsec:allocation_protocol} effectively maximizes tracking performance within a given market interval. However, sustained multi-hour participation in frequency regulation introduces an additional challenge of preserving battery SoC across consecutive intervals. Since the terminal SoC of one interval directly determines the admissible operating region of the next, cumulative energy deviations can lead to SoC drift and, ultimately, battery saturation and limit the battery participation in frequency regulation. Addressing this inter-temporal coupling is therefore essential for reliable long-horizon HES participation in regulation markets. While Reg-D signals are designed to be energy-neutral over extended horizons, short-term stochastic imbalances frequently cause the SoC to drift from its initial setpoint $E_{0}$ or designated target band (Fig.~\ref{fig:SoC_opti}). Significant deviations can drive the SoC to its operational limits ($\underline{E}$ or $\overline{E}$), compromising the HES's ability to provide continuous service. 

To mitigate this risk and ensure sustained market participation, we propose a real-time adaptive SoC drift correction approach. For the allocated battery power $P^{*}_{\mathrm{batt}}[k]$, this approach dynamically regulates the maximum battery power whenever the SoC enters a defined buffer region ($\Delta_{\text{b}}$) near its upper ($E_{\text{u}}$) or lower ($E_{\text{l}}$) bounds. The buffer is defined relative to the target band, e.g., $\Delta_{\text{b}}$ = 0.1 ($E_{\text{u}}$-$E_{\text{l}}$). The complete drift correction approach, which determines the final battery dispatch for each time step $k\in\mathcal{K}$, is given by Algorithm~\ref{alg:hes_dispatch_short}.

\begin{algorithm}[tb]
\caption{Proposed Real-time SoC Drift Correction Approach}
\label{alg:hes_dispatch_short}
\begin{algorithmic}[1]
\renewcommand{\algorithmicrequire}{\textbf{Given:}}
\Require  
    $P^{*}_{\mathrm{batt}}[k]$, $\overline{P}_{\mathrm{batt}}$, 
    $\eta_{\mathrm{c}}$, $\eta_{\mathrm{d}}$, 
    $E[k]$, $E_{0}$, $E_{\text{u}}$,$E_{\text{l}}$,$\Delta_{\text{b}}$, $\Delta t$,
\State $P_{\mathrm{max}} \gets \overline{P}_{\mathrm{batt}}$

\If{$E[k] > E_{\text{u}}-\Delta_{\text{b}}$}
    \State $k_{u} \gets \max\left\{0,\; \dfrac{E_{\text{u}} - E[k]}{\Delta_{\text{b}}}\right\}$
    \If{$P^{*}_{\mathrm{batt}}[k] < 0$}  
        \State $P_{\mathrm{max}} \gets k_{u}\,\overline{P}_{\mathrm{batt}}$
    \EndIf
\ElsIf{$E[k] < E_{l} + \Delta_{\mathrm{b}}$}
    \State $k_{l} \gets \max\!\left\{0,\; \dfrac{E[k] - E_{l}}{\Delta_{\mathrm{b}}}\right\}$
    \If{$P^{*}_{\mathrm{batt}}[k] > 0$}  
        \State $P_{\mathrm{max}} \gets k_{l}\,\overline{P}_{\mathrm{batt}}$
    \EndIf
\EndIf
\If{$P^{*}_{\mathrm{batt}}[k] \ge 0$}  
    \State $P_{\mathrm{batt}}[k] \gets \min\{P^{*}_{\mathrm{batt}}[k],\; P_{\mathrm{max}}\}$
    \State $\Delta E_{k} \gets \dfrac{P_{\mathrm{batt}}[k]}{\eta_{\mathrm{d}}}\dfrac{\Delta t}{E_{\text{c}}}$
\Else
    \State $P_{\mathrm{batt}}[k] \gets \max\{P^{*}_{\mathrm{batt}}[k],\; -P_{\mathrm{max}}\}$
    \State $\Delta E_{k} \gets P_{\mathrm{batt}}[k]\;\eta_{\mathrm{c}}\dfrac{\Delta t}{E_{\text{c}}}$
\EndIf

\State $E[k{+}1] \gets E[k] - \Delta E_{k}$

\end{algorithmic}
\end{algorithm}

\begin{figure}[!htb]
\centering
\subfloat[Regulation signal and HES power]{%
    \includegraphics[width=0.5\textwidth]{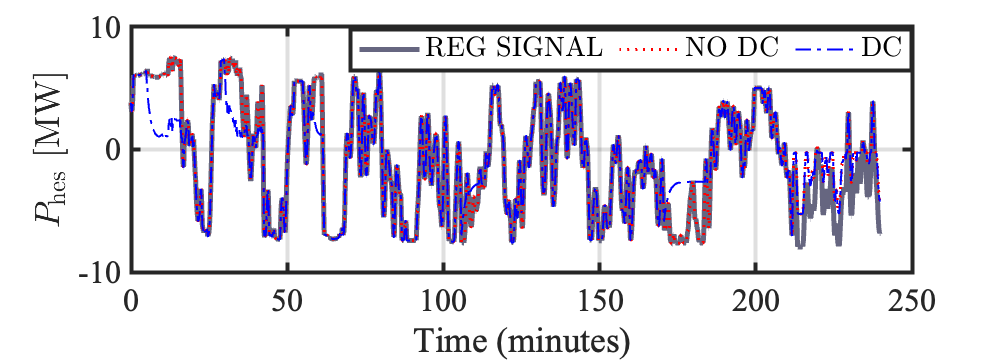}
    \label{fig:HES_power_drift}
}
\vspace{1em}
\subfloat[Battery net power dispatch $P_{\text{batt}}$]{%
    \includegraphics[width=0.5\textwidth]{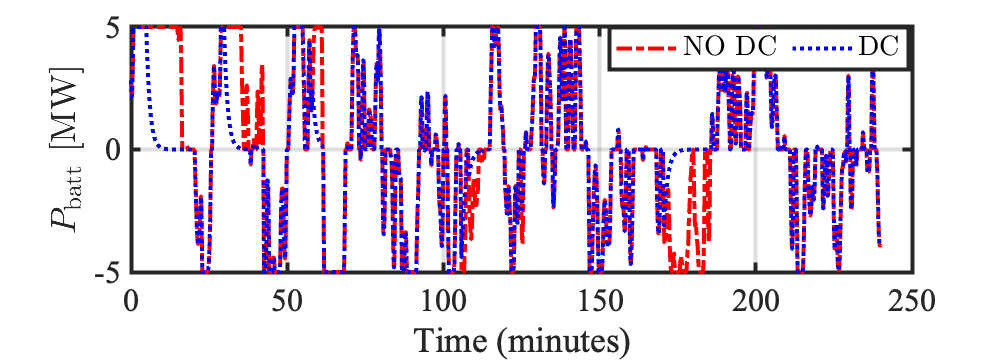}
    \label{fig:Battery_dispatch_drift}
}
\hfill
\subfloat[Battery SoC $E$]{%
    \includegraphics[width=0.5\textwidth]{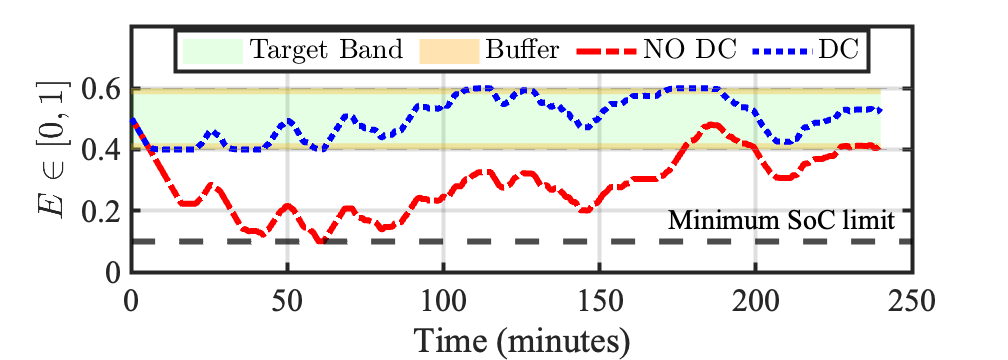}
    \label{fig:SoC_opti_drift}
}
\caption{Comparison of real-time power allocation rule based with (DC) and without SoC drift correction(NO DC) dispatch for scenario $S_{\text{1}}$ with $C$~=~6.5 MW. The plots show the net regulation signal $C\mathbf{r}$, $P_{\text{hes}}$, $P_{\text{batt}}$,and battery SoC $E$. The upper bound $E_{\text{u}}$, lower bound $E_{\text{u}}$ and buffer zone $\Delta_{b}$ are 0.6, 0.4 and 0.02 respectively.}
\label{fig:Drift_correction}
\end{figure}

Fig.~\ref{fig:Drift_correction} shows the efficacy of the real-time power allocation with and without SoC drift correction under scenario~$S_{1}$ (Max Flex – No PV Curtailment) from Table~ \ref{tab:comparison} over a consecutive four-hour horizon. The target band was defined by upper bound $E_{\text{u}}~=~0.6$, lower bound $E_{\text{l}}~=~0.4$ and buffer zone $\Delta_{b}=0.02$. As observed, the baseline strategy without drift correction (``NO DC'') fails to manage the cumulative energy imbalance, causing the battery to hit its minimum operational limit $\underline{E}$~=~0.1. However, the proposed strategy (``DC") detects when the SoC breaches the buffer zone and regulates the dispatch to constrain the SoC within the target region. We define an equivalent performance score $\tilde{x}_{\text{p}}$, which is same as \eqref{xp} but for four consecutive hour. The $\tilde{x}_{\text{p}}$~=~0.97 for the baseline strategy without drift correction, while $\tilde{x}_{\text{p}}$~=~0.89 incorporating the drift correction. Although this strategy results in a slightly lower performance score but prevents SoC operational limit violations, ensuring the sustained availability.

\section{Conclusion}
\label{Conclusion} 
This paper characterizes the notion of maximum power flexibility for HES composed of solar PV, battery storage, and a controllable load. A flexibility envelope was derived to determine the largest up-down regulation range that the HES can reliably provide at each time step $k\in\mathcal{K}$, taking into account operational constraints and PV curtailment. The framework enables real-time power allocation and is designed to support capacity bidding in the PJM's frequency regulation market. The maximality and symmetry of the proposed flexibility envelope were analytically established. Multiple operational scenarios were examined as summarized in Table~\ref{tab:comparison}, including PV curtailment and purely renewable-driven load, each resulting in a different nominal trajectory and flexibility range. The role of PV curtailment in increasing dispatchable range was contrasted with its trade-off against full PV power utilization.
 
This real-time power allocation protocol was benchmarked against an optimization framework to validate optimality of the dispatch. This optimization framework assumes perfect solar generation forecast and regulation signal. While the optimization provides a performance upper bound, the proposed rule based allocation protocol achieves same results with reduced computational overhead, making it suitable for real-time implementation. Further, By integrating solar forecasts directly into the bidding, the framework dynamically sizes the symmetric regulation capacity the HES can offer. This analysis demonstrates that HES can successfully participate in regulation markets even with high solar variability. 
To ensure sustainable HES participation over multi-hour operation, an adaptive SoC drift correction approach is proposed to prevent battery saturation, thereby maintaining operational continuity and enabling practical long-term HES deployment in frequency regulation markets. Future work will advance the proposed framework by developing robust bidding strategies that explicitly internalize the stochastic nature of regulation signal, mitigating the financial risks associated with extreme forecast errors. Further, this expansion will also incorporate heterogeneous ramp-rate limits of different HES assets. Finally, these enhancements will be validated across broader asset classes with complex dynamics, including wind generation, data centers and aggregated EV charging stations as controllable loads.
\section*{Acknowledgment}
The authors thank Dakota Hamilton at the University of Vermont for insightful discussions on flexibility of hybrid energy systems.

\vspace{-0.5cm}
\bibliographystyle{IEEEtran}
\bibliography{Ref}

@ARTICLE{Gnibga2024,
  author={Gnibga, Wedan Emmanuel and Blavette, Anne and Orgerie, Anne-Cécile},
  journal={IEEE Transactions on Sustainable Computing}, 
  title={Renewable Energy in Data Centers: The Dilemma of Electrical Grid Dependency and Autonomy Costs}, 
  year={2024},
  volume={9},
  number={3},
  pages={315-328},
  doi={10.1109/TSUSC.2023.3307790}}

@misc{mishra2025,
      title={Optimal Bidding and Coordinated Dispatch of Hybrid Energy Systems in Regulation Markets}, 
      author={Tanmay Mishra and Dakota Hamilton and Mads R. Almassalkhi},
      year={2025},
      eprint={2510.26602},
      archivePrefix={arXiv},
      primaryClass={eess.SY},
      url={https://arxiv.org/abs/2510.26602}, 
}

@ARTICLE{Wang2016,
  author={Wang, Hao and Huang, Jianwei and Lin, Xiaojun and Mohsenian-Rad, Hamed},
  journal={IEEE Transactions on Smart Grid}, 
  title={Proactive Demand Response for Data Centers: A Win-Win Solution}, 
  year={2016},
  volume={7},
  number={3},
  pages={1584-1596},
  doi={10.1109/TSG.2015.2501808}}

@ARTICLE{Karfopoulos2015,
  author={Karfopoulos, Evangelos L. and Panourgias, Kostas A. and Hatziargyriou, Nikos D.},
  journal={IEEE Transactions on Power Systems}, 
  title={Distributed Coordination of Electric Vehicles providing V2G Regulation Services}, 
  year={2016},
  volume={31},
  number={4},
  pages={2834-2846},
  doi={10.1109/TPWRS.2015.2472957}}

@manual{pjm2025manual12,
  title        = {PJM Manual 12: Operating Agreement Accounting},
  author       = {{PJM Interconnection}},
  year         = {2025},
  month        = {June},
  edition      = {Revision 55},
  url          = {}
}

@techreport{ANL_survey_2016,
  author       = {Zhou, Zhi and Levin, Todd and Conzelmann, Guenter},
  title        = {Survey of U.S. Ancillary Services Markets},
  institution  = {Argonne National Lab. (ANL), IL (USA)},
  doi          = {10.2172/1236451},
  url          = {},
  place        = {},
  year         = {2016},
  month        = {01}}

@INPROCEEDINGS{Bolun2016,
  author={Bolun Xu and Dvorkin, Yury and Kirschen, Daniel S. and Silva-Monroy, C. A. and Watson, Jean-Paul},
  booktitle={2016 IEEE Power and Energy Society General Meeting (PESGM)}, 
  title={A comparison of policies on the participation of storage in U.S. frequency regulation markets}, 
  year={2016},
  volume={},
  number={},
  pages={1-5},
  doi={10.1109/PESGM.2016.7741531}}

@INPROCEEDINGS{Mazen2023,
  author={Elsaadany, Mazen and Almassalkhi, Mads R.},
  booktitle={2023 62nd IEEE Conference on Decision and Control (CDC)}, 
  title={Battery Optimization for Power Systems: Feasibility and Optimality}, 
  year={2023},
  volume={},
  number={},
  pages={562-569},
  doi={10.1109/CDC49753.2023.10384282}}

@ARTICLE{Dang2014,
  author={Dang, Jie and Seuss, John and Suneja, Luv and Harley, Ronald G.},
  journal={IEEE Journal of Emerging and Selected Topics in Power Electronics}, 
  title={SoC Feedback Control for Wind and ESS Hybrid Power System Frequency Regulation}, 
  year={2014},
  volume={2},
  number={1},
  pages={79-86},
  doi={10.1109/JESTPE.2013.2289991}}

@ARTICLE{Garcia2021,
  author={Garcia-Torres, Felix and Bordons, Carlos and Tobajas, Javier and Real-Calvo, Rafael and Santiago, Isabel and Grieu, Stephane},
  journal={IEEE Transactions on Power Systems}, 
  title={Stochastic Optimization of Microgrids With Hybrid Energy Storage Systems for Grid Flexibility Services Considering Energy Forecast Uncertainties}, 
  year={2021},
  volume={36},
  number={6},
  pages={5537-5547},
  doi={10.1109/TPWRS.2021.3071867}}

@ARTICLE{LiHao2018,
  author={Li, Hao and Lu, Zongxiang and Qiao, Ying and Zhang, Baosen and Lin, Yisha},
  journal={IEEE Transactions on Power Systems}, 
  title={The Flexibility Test System for Studies of Variable Renewable Energy Resources}, 
  year={2021},
  volume={36},
  number={2},
  pages={1526-1536},
  doi={10.1109/TPWRS.2020.3019983}}

@ARTICLE{Bolun2018,
  author={Xu, Bolun and Shi, Yuanyuan and Kirschen, Daniel S. and Zhang, Baosen},
  journal={IEEE Transactions on Power Systems}, 
  title={Optimal Battery Participation in Frequency Regulation Markets}, 
  year={2018},
  volume={33},
  number={6},
  pages={6715-6725},
  doi={10.1109/TPWRS.2018.2846774}}

@ARTICLE{Calero2022,
  author={Calero, Ivan and Cañizares, Claudio A. and Bhattacharya, Kankar and Baldick, Ross},
  journal={IEEE Transactions on Smart Grid}, 
  title={Duck-Curve Mitigation in Power Grids With High Penetration of PV Generation}, 
  year={2022},
  volume={13},
  number={1},
  pages={314-329},
  doi={10.1109/TSG.2021.3122398}}

@article{oshaughnessy2020,
  title={Too much of a good thing? Global trends in the curtailment of solar PV},
  author={O'Shaughnessy, Eric and Heeter, Jenny and Fu, Ran and Ardani, Kristen and Margolis, Robert},
  journal={Solar Energy},
  volume={210},
  pages={632--641},
  year={2020},
  publisher={Elsevier},
  doi={10.1016/j.solener.2020.08.056}
}

@article{BIALEK2020,
title = {What does the {GB} power outage on 9 August 2019 tell us about the current state of decarbonised power systems?},
journal = {Energy Policy},
volume = {146},
pages = {111821},
year = {2020},
issn = {0301-4215},
doi = {https://doi.org/10.1016/j.enpol.2020.111821},
author = {Janusz Bialek},
}

@ARTICLE{Sovana2024,
  author={Qudaih, Sovana and Bektas, Zeynep and Guven, Denizhan and Kayakutlu, Gülgün and Kayalica, Mehmet Ozgur},
  journal={IEEE Transactions on Engineering Management}, 
  title={Technology Assessment of Hydrogen Storage: Cases Enabling the Clean Energy Transition}, 
  year={2024},
  volume={71},
  number={},
  pages={5744-5756},
  doi={10.1109/TEM.2024.3366973}}

@ARTICLE{Brahma2021,
  author={Brahma, Sarnaduti and Nazir, Nawaf and Ossareh, Hamid and Almassalkhi, Mads R.},
  journal={IEEE Transactions on Power Systems}, 
  title={Optimal and Resilient Coordination of Virtual Batteries in Distribution Feeders}, 
  year={2021},
  volume={36},
  number={4},
  pages={2841-2854},
  doi={10.1109/TPWRS.2020.3043632}}

@ARTICLE{Brahma2023,
  author={Brahma, Sarnaduti and Khurram, Adil and Ossareh, Hamid and Almassalkhi, Mads},
  journal={IEEE Transactions on Smart Grid}, 
  title={Optimal Frequency Regulation Using Packetized Energy Management}, 
  year={2023},
  volume={14},
  number={1},
  pages={341-353},
  doi={10.1109/TSG.2022.3197703}}

@ARTICLE{Xie2020,
  author={Xie, Yijing and Zhang, Yichen and Lee, Wei-Jen and Lin, Zongli and Shamash, Yacov A.},
  journal={IEEE/CAA Journal of Automatica Sinica}, 
  title={Virtual Power Plants for Grid Resilience: A Concise Overview of Research and Applications}, 
  year={2024},
  volume={11},
  number={2},
  pages={329-343},
}

@article{palensky2011demand,
  title={Demand side management: Demand response, intelligent energy systems, and smart loads},
  author={Palensky, Peter and Dietrich, Dietmar},
  journal={IEEE Transactions on Industrial Informatics},
  volume={7},
  number={3},
  pages={381--388},
  year={2011},
  publisher={IEEE}
}

@ARTICLE{Nosair2015,
  author={Nosair, Hussam and Bouffard, Franois},
  journal={IEEE Transactions on Sustainable Energy}, 
  title={Flexibility Envelopes for Power System Operational Planning}, 
  year={2015},
  volume={6},
  number={3},
  pages={800-809},
  doi={10.1109/TSTE.2015.2410760}}

@article{Tanmay2024,
	author = {Tanmay Mishra and Amritanshu Pandey and Mads R. Almassalkhi},
	issn = {0378-7796},
	journal = {Electric Power Systems Research},
	pages = {110767},
	title = {Predictive optimization of hybrid energy systems with temperature dependency},
	year = {2024},
}

@ARTICLE{Mohandes2019,
  author={Mohandes, Baraa and Moursi, Mohamed Shawky El and Hatziargyriou, Nikos and Khatib, Sameh El},
  journal={IEEE Trans on Power Systems}, 
  title={A Review of Power System Flexibility With High Penetration of Renewables}, 
  year={2019},
  volume={34},
  number={4},
  pages={3140-3155}}

@ARTICLE{Makarov2009,
  author={Makarov, Yuri V. and Loutan, Clyde and Ma, Jian and de Mello, Phillip},
  journal={IEEE Transactions on Power Systems}, 
  title={Operational Impacts of Wind Generation on California Power Systems}, 
  year={2009},
  volume={24},
  number={2},
  pages={1039-1050}
}

@Article{Rosewater2019,
  author  = {Rosewater, David M. and Copp, David A. and Nguyen, Tu A. and Byrne, Raymond H. and Santoso, Surya},
  journal = {IEEE Access},
  title   = {Battery Energy Storage Models for Optimal Control},
  year    = {2019},
  pages   = {178357-178391},
  volume  = {7},
  doi     = {10.1109/ACCESS.2019.2957698},
}

@inbook{Masters2004,
author = {Gilbert M. Masters},
publisher = {John Wiley \& Sons, Ltd},
isbn = {9780471668824},
title = {``Renewable and Efficient Electric Power Systems"},
chapter = {8},
pages = {445-504},
year = {2004},
}

@inproceedings{Arash2025,
  title     = {Towards Input-Convex Neural Network Modeling for Battery Optimization in Power Systems},
  author    = {Arash Omidi and Tanmay Mishra and Mads R. Almassalkhi},
  booktitle = {2025 American Control Conference (ACC)},
  year      = {2025},
}

@ARTICLE{Ferreira2019,
  author={Rosewater, David and Ferreira, Summer and Schoenwald, David and Hawkins, Jonathan and Santoso, Surya},
  journal={IEEE Transactions on Smart Grid}, 
  title={Battery Energy Storage State-of-Charge Forecasting: Models, Optimization, and Accuracy}, 
  year={2019},
  volume={10},
  number={3},
  pages={2453-2462},
  doi={10.1109/TSG.2018.2798165}}

@TechReport{Murphy2021,
  author   = {Caitlin Murphy and Andrew Mills},
  title    = {Hybrid Energy Systems: Opportunities for Coordinated Research},
  year     = {2021},
}
\end{document}